\documentclass[a4paper,12pt]{article} 

\usepackage[a4paper, left=25mm, right=25mm, top=35mm, bottom=35mm]{geometry}
\usepackage[utf8]{inputenc}

\usepackage[intlimits]{amsmath}

\usepackage[english]{babel}
\usepackage[T2A]{fontenc}
\usepackage{indentfirst}

\usepackage{tikz}
\usepackage{mathrsfs}
\usepackage{mathtools}
\usetikzlibrary{arrows}
\usepackage{amssymb}
\usepackage{amsthm}
\usepackage{bm}

\usepackage{listings}
\usepackage{color}
\usepackage{graphicx}
\usepackage{psfrag}
\usepackage{wrapfig}
\usepackage{caption}
\newcommand{\arr}{\longleftrightarrow}

\renewcommand{\L}{\Lambda}
\renewcommand{\a}{\alpha}

\newcommand{\g}{\gamma}

\newcommand{\e}{\varepsilon}

\newcommand{\N}{\mathbb{N}}
\newcommand{\Z}{\mathbb{Z}}
\renewcommand{\H}{\mathbb{H}}
\renewcommand{\U}{\mathbb{U}}
\newcommand{\R}{\mathbb{R}}
\newcommand{\one}{\mathbb{I}}

\newcommand{\mS}{\mathcal{S}}
\newcommand{\ovphi}{\overline \varphi}

\renewcommand{\d}{\partial}
\newcommand{\hs}{\hat{\sigma}}

\newcommand{\be}{\begin{equation}}
\newcommand{\ee}{\end{equation}}
\newcommand{\benn}{\begin{equation*}}
\newcommand{\eenn}{\end{equation*}}
\newcommand{\baln}{\begin{align*}}
\newcommand{\ealn}{\end{align*}}
\newcommand{\bal}{\begin{align}}
\newcommand{\eal}{\end{align}}

\newcommand{\bee}{\begin{eqnarray}}
\newcommand{\eee}{\end{eqnarray}}

\newcommand{\beenn}{\begin{eqnarray*}}
\newcommand{\eeenn}{\end{eqnarray*}}

\addcontentsline{toc}{section}{Bibliography}

\newtheorem{theorem}{Theorem}

\newtheorem{lemma}[theorem]{Lemma}

\newtheorem{cor}[theorem]{Corollary}
\newtheorem{prop}[theorem]{Proposition}

\title{Another proof of $p_{c} = \tfrac{\sqrt{q}}{1+\sqrt{q}}$ on $\Z^2$ with $q \in [1,4]$ for random-cluster model.} 
\author{E.~Mukoseeva \thanks{ Scuola Internazionale Superiore di Studi Avanzati (SISSA)}
\and D.~Smirnova  \thanks{Universit\'e de Gen\`eve}}
\date{\today} 

\begin{document}

\maketitle

\begin{abstract}
In this paper we give another proof of $p_{c} = \tfrac{\sqrt{q}}{1+\sqrt{q}}$ for random-cluster model in $\Z^2$ for $q \in [1,4]$, based on the method of parafermionic observables.
\end{abstract}

\section{Introduction}


\par Consider a finite connected subgraph $\Omega$ of $\Z^2$ and denote the set of its edges by $E(\Omega)$, the set of its vertices by $V(\Omega)$ and  the set of its faces by $F(\Omega)$. 
Let us also define its boundary vertices $\d \Omega$ to be the set of vertices of $\Omega$ with one neighbour in $\Z^2 \backslash \Omega$.   
An edge is called a {\it{boundary edge}} if its ends belong to $\d \Omega$, otherwise, the edge is called an {\it{inner edge}}. 
Let us define a {\it{configuration}} $\omega$ as an element of the set $\{0,1\}^{E(\Omega)}$. An edge $e \in E(\Omega)$ is called {\it{open}} in $\omega$ if $\omega(e)=1$, it is called {\it{closed}} otherwise. 
We say that two points $x , y \in V(\Omega)$ are connected in $\omega$ if there is a path of open edges connecting them, this event is denoted by $x \arr y$. If we are allowed to use only open edges from $A \subset \Omega$, we use the notation  $x \underset{A} \arr y$.
A maximal set of vertices connected to each other is called a {\it{cluster}}.

\par The random-cluster model is a family of probability measures on  $\{0,1\}^{E(\Omega)}$ introduced in 1969 by K.~Fortuin and P.~Kasteleyn \cite{FK72, F272, F372} as a two-parametric generalisation of Bernoulli percolation.
It depends on two parameters $p \in [0,1]$ and $q \in \R^+$ and defined as
\be \label{def}
\phi_{\Omega, p ,q}^ \# (\omega) = \frac{1}{Z^\#_{\Omega, p, q}} q^{k^{\#}(\omega)} \prod_{e \in E(\Omega)}	p^{\omega(e)}(1-p)^{1-\omega(e)},
\ee 
where the normalisation coefficient 
\be \label{Z}
Z_{\Omega, p, q}^{\#} = \sum_{w \in \{0,1\}^{E(\Omega)}} q^{k^{\#}(\omega)} \prod_{e \in E(\Omega)}	p^{\omega(e)}(1-p)^{1-\omega(e)}
\ee
is called the {\it{partition function}}. The notation $k^{\#}(\omega)$ is defined as follows:
if $\# = 0$, then $k^0(\omega)$ is the number of clusters in $\omega$ and corresponds to the measure $\phi^0_{\Omega, p ,q} (\omega) $ (in this case, we speak of the random-cluster measure with {\it{free boundary conditions}}), 
$k^1(\omega)$ is the number of clusters in $\omega$ when all boundary vertices belong to one cluster. Then, all clusters touching the boundary are counted as one. This measure is denoted $\phi^1_{\Omega, p ,q} (\omega)$ and called the random-cluster measure with {\it{wired boundary conditions}}.
It is also possible to define the random-cluster measure for any intermediate boundary conditions $\xi$. 
\par We can define random-cluster measures 
on an infinite graph by taking the limit of measures 
$\lim_{i  \rightarrow \infty} \phi^0_{G_i,p,q}$ (resp. $\lim_{i \rightarrow \infty} \phi^1_{G_i,p,q}$) where $G_i$ is an increasing sequence of finite graphs converging to the infinite graph (for example boxes of increasing size).

When $q = 1$, all measures are equal to the product measure used in Bernoulli percolation. 
When $q$ tends to zero, the random-cluster measure converges to uniform measures on connected subgraphs, spanning trees or spanning forests depending on the behaviour of $p$ \cite{Gr06}.

\par An important property of the random-cluster model is its connection to Ising and Potts models.
The Edward-Sokal coupling \cite{ES88} is the most standard way to relate random-cluster model for some $p$ and integer $q \ge 2$ with the Potts model with $q$ colours and the inverse temperature $\beta$ such that $p = 1 - e^{-\beta}$, it was done by constructing a correspondence between their configurations.
Many results about Ising and Potts models were found using this construction.

\par As Bernoulli percolation and the Ising and Potts models, the random-cluster model undergoes a phase transition. 
For any value of $q$, there exists $p_{c} = p_{c}(q) \in (0,1)$ such that for $p > p_{c}$ there is almost surely an infinite connected component for a configuration on $\Z^2$ with arbitrary boundary conditions.
If $p<p_{c}$ there is no infinite connected component.
\par The exact value of $p_{c}(q)$ was one of the most important questions not only for the random-cluster model itself, but also for Potts and Ising models, where it gives
the exact value of inverse critical temperature $\beta_{cr}$.
\par It was proven in 2012 by H.~Duminil-Copin and V.~Beffara \cite{BDC12} that $p_{c}$ is equal to the self-dual point, which gives the relation $p = p^* = p_{c}$ between $p$ and the parameter $p^*$ on the corresponding dual graph (see Section \ref{dual}). More precisely,

\begin{theorem}
\label{Th0}
On $\Z^2$, we have that for every $q \ge 1$,
\be 
p_{c}(q) = p_{sd}(q) =  \frac{\sqrt{q}}{1+\sqrt{q}}.
\ee
\end{theorem}


\par Both \cite{BDC12} and alternative proofs of Theorem \ref{Th0} given in $\cite{DCM14}$ and \cite{DCRT17}
use the crossing probabilities, i.e the probability of the rectangle of size $n \times \a n$ to be crossed in a horizontal direction. By generalisation of the Russo-Seymour-Welsh approach, this probability is bounded away from $0$ and $1$ independently on $n$. 
Also, both methods are strongly based on the FKG inequality, which bounds from below the probability for two different crossings to hold simultaneously by the product of probabilities for each of the crossings to hold independently of another one (see Section \ref{FKGdef}).


The approach used here
does not use crossing probabilities and barely uses the FKG inequality. 
It is based on the method of parafermionic observables (see Section \ref{paraf}). 
This approach was used in other lattice models to deduce the critical point
(for self-avoiding walks in \cite{DCS10} and \cite{Gl14}, for loop $O(n)$ models in \cite{DCGPS17}). 
For the random-cluster model, this method was used in the case $q \ge 4$ in \cite{BDCS15}. 
Also this method was applied in \cite{Lis17} in a general case of Ising model with no translational invariance.
Despite all this progress, the application of this method in the case of $q < 4$ remained open until now.

The key idea of the proof is to bound the derivative of the probability to have a path from the middle of the box to the boundary or between two fixed points on the opposite boundary of the strip.
That is why this approach can be used in the case where crossing probabilities cannot be bounded, or even when FKG does not hold.
Also we hope that this technique could be useful to study models without invariance under translation.

The paper is organised as follows.
Section \ref{Not} gives all definitions and statements required for the proof.
In Section \ref{Sketch3}, Theorem \ref{Th0} is proven in the simple case $q \in [1,3]$.
The idea of the proof for $q \in (3,4]$ is slightly similar to the previous one, but requires an additional construction. The strategy of the proof is given in Section \ref{Sketch4}, and the proof is completed in Sections \ref{Proofthlem}--\ref{twoc}.

\paragraph{Acknowledgments}  
The authors thank Hugo Duminil-Copin for suggesting the problem and many fruitful discussions.
The work was supported by the NCCR Swissmap.

\section{Definitions and basic properties} \label{Not}


\subsection{Domain Markov property and finite energy property}
The Domain Markov property states that the only influence that a configuration outside a graph $G$ has on a configuration inside $G$ is via induced boundary conditions.
More formally, suppose we have $G \subset G'$ and some boundary conditions $\xi'$ on $G'$ and let $\psi$ be a configuration on $G'\backslash G$. 
Let us define the following boundary conditions: two boundary clusters are counted as one if they are connected in $\psi$ (including the connections in $\xi'$). These boundary conditions are called {\it{induced}} by the configuration $\psi$ with boundary conditions $\xi'$ and denoted $(\psi, \xi')$. 
Then, for any configuration $\omega$ on $G'$ and its restriction to $G$ denoted $\omega_{|G}$, the following property holds:
\benn
\phi_{G', p,q}^{\xi'} (\omega\, |\, \omega_e = \psi_e  \text{ for all } e \in G'\backslash G) = \phi_{G, p,q}^{(\psi,\xi')} (\omega_{|G}).
\eenn

The finite energy property allows one to compare configurations which differ only on finitely many edges. 
For any $\e \in (0,\tfrac{1}{2})$ and for any $q>1$, there exists a constant $c = c(\e, q)>0$ such that, for every configuration $\omega$ and $\omega'$ differing only on $k$ edges 
\benn
c^{-k} \le \frac{\phi^\xi_{G,p,q}(\omega)}{\phi^\xi_{G,p,q}(\omega')}\le c^{k}
\eenn
for any $G$, any boundary conditions $\xi$, any fixed $q$ and any $p \in (\e, 1-\e)$.
\par Let us use the notation $\omega^e$ (resp. $\omega_e$) for the configuration obtained from the configuration $\omega$ by setting the edge $e$ to be open (resp. close).  
Then for a fixed value of $p$, this property can also be written in the following way:
\be \label{alltoclosed}
\frac{p}{q(1-p)} \le \frac{\phi_{G,p,q}^\xi(\omega^e)+\phi_{G,p,q}^\xi(\omega_e)}{\phi_{G,p,q}^\xi(\omega_e)} \le \frac{p}{(1-p)}
\ee
for any $G, p, \xi$ and any $q \ge 1$.

\subsection{Increasing events} \label{FKGdef}

\par The event $A$ is called {\it{increasing}} if for any edge $e$ and any configuration $\omega$ the fact that $A$ holds for $\omega_e$ implies that $A$ holds for $\omega^e$. 
If $A$ and $B$ are two increasing events, then the following inequality, called FKG inequality \cite{FKG71}, holds for any values of $p\in [0,1]$, $ q \ge 1$ and for any graph $G$ and boundary condition $\xi$: 
\be
\phi^\xi_{G, p, q} (A \cap  B) \ge \phi^\xi_{G, p, q} (A) \phi^\xi_{G, p, q} (B).
\ee

\par Let us now turn to the comparison between different boundary conditions. Write $\xi \ge \chi$ if any two boundary points seen as parts of one cluster in $\chi$ have the same property in $\xi$. Then, for any increasing $A$ and any $G, p$ and $q \ge1$,
\be \label{boundcond}
\phi^\xi_{G, p,q} (A)\ge \phi^\chi_{G, p,q}(A).
\ee
This inequality, together with the definition of induced boundary conditions, can be used to compare the measures in different domains. For example, if $G \subset G'$ and $A$ is defined on $G$, then
\be \label{boundcond1}
\phi^1_{G,p,q} (A) \ge \phi^\xi_{G',p,q} (A) \ge\phi^0_{G,p,q} (A)
\ee
for any $p, q \ge1$ and any $\xi$.
\par The edge $e$ is called {\it{pivotal}} for the event $A$ in a given configuration $\omega$ if $A$ holds for $\omega^e$ and does not hold for $\omega_e$. Note that the fact that the edge is pivotal does not depend on the state of this edge in $\omega$.
The following inequality
holds for any increasing event $A$, boundary conditions $\xi$ and graph $G$ \cite{Gr06}: 
\be \label{difineq}
\frac{d}{dp} \phi^\xi_{G,p,q}(A) \ge c \sum_{e \in E(G)} \phi^\xi_{G,p,q}( e \text{ is pivotal for } A),
\ee
where $c$ does not depend on $G$, $A$ or $p$. 

\subsection{Duality} \label{dual}

Let us define the {\it{dual lattice}} $L^*$ of a planar lattice $L$ as follows.
The set of vertices of $L^*$ corresponds to the set $F(L)$ of faces of $L$ including the infinite face. In words, we put a vertex of the dual lattice in the middle of each face of $L$. 
Then, we connect all pairs of vertices corresponding to adjacent faces. 
Each edge of the dual lattice intersects exactly one edge of the primal lattice. 
\par For a finite graph $G \subset L$, let $G^*$ denote the subgraph of $L^*$ with edges corresponding to edges of $G$ and vertices to the endpoints of these edges.
\par The configuration $\omega^*$ on $G^*$ {\it{dual}} to a configuration $\omega$ on $G$ is defined as follows:
\benn
\omega^*(e^*) = 1 - \omega(e),
\eenn
where $e$ is the edge of the primal lattice intersecting $e^*$. The probability of $\omega^*$ is set to be the probability of $\omega$.
The event that $x$ is connected to $y$ by a dual-open path is denoted by $x \overset{*}{\arr} y$.

There exists a value $p^* = p^*(p,q)  \in [0,1]$ such that the following relation is true \cite{Gr06, DC13}:
\be
\phi^\xi_{G, p,q} (\omega) = \phi^{\xi^*}_{G^*, p^*,q}(\omega^*). 
\ee 
Here, $\xi^*$ denotes boundary conditions dual to $\xi$, for example wired boundary conditions are dual to free boundary conditions and vice versa.
The value of $p^*$ is related to $p$ via the following equation:
\be
\frac{p p^*}{(1-p)(1-p^*)} = q.
\ee
As mentioned above, there exists the unique point $p_{sd}$ such that $p_{sd} = (p_{sd})^*$. Its value is equal to 
\benn
p_{sd} = \frac{\sqrt{q}}{1+\sqrt{q}}.
\eenn

\subsection{Parafermionic observable} \label{paraf}
\par The {\it{medial lattice}} of a lattice $L$ is denoted by $L^\diamond$ and defined as follows.
The set $V(L^\diamond)$ corresponds to $E(L)$ (that can be interpreted as if we put a vertex in the middle of every edge of $L$). 
There is an edge between two vertices of $L^\diamond$ if the corresponding edges have one common end vertex and are adjacent to the same face.
Note that the faces of $L^\diamond$ correspond to $V(L)\cup V(L^*)$. We orient the edges of $L^\diamond$ counterclockwise around faces corresponding to vertices of $L$ (and, thus, clockwise around faces corresponding to vertices of $L^*$).
Note that the lattice $(\Z^2)^\diamond$ is a rescaled and rotated version of $\Z^2$.
 
\par The {\it{medial graph}} $G^\diamond$ of a finite graph $G \subset L$ is the subgraph of $L^\diamond$ made of all faces corresponding to the vertices of $G$ and $G^*$, all edges surrounding these faces and all vertices incident to these edges.

\par Let $\gamma$ be a path on $G^\diamond$ and let $e$ and $e'$ be two edges belonging to $\gamma$. Then, the {\it{winding}} $W_\gamma(e,e')$ is the total rotation done by $\gamma$ on the way from the middle of $e$ to the middle of $e'$. If $e$ or $e'$ do not belong to $\gamma$, we set $W_\gamma(e,e')=0$. 
\par Let us pick two vertices $a$ and $b$ on the boundary of $G$ and define {\it{Dobrushin boundary conditions}}, denoted $(a,b)$, as wired on the boundary arc from $a$ to $b$ clockwise and free for the rest of $\d G$. Let us note that for dual configuration $(a,b)^* = (b,a)$.
There is a primal cluster attached to the wired boundary part and a dual cluster attached to the free boundary part (in the dual configuration, it corresponds to a wired arc). The curve on the medial lattice going between these two clusters is called an {\it{exploration path}} (for more details, see \cite{DC13}). Note that it is oriented. Let us also add two edges $e_a$ and $e_b$ to begin and to end $\gamma$ outside of $G$, these edges are oriented according to the orientation rules stated above.

\par We follow \cite{Sm10} and \cite{DCST17} to define the {\it{parafermionic observable}}:
\benn
\hat{F}(e) = \hat{F}_{G,a,b,q,p_{sd}}(e) = \phi_{G,p_{sd},q}^{(a,b)}(e^{i \hs W_\gamma(e,e_b)} \one_{e \in \gamma}) \quad \forall e \in E(G^\diamond),
\eenn
where $\gamma$ is the exploration path in $\omega$ and $\hs$ satisfies the following relation:
\benn
\cos(\tfrac{\hs \pi}{2}) = \tfrac{\sqrt{q}}{2}.
\eenn

\par Let us take any set $V \in V(G^\diamond)$ such that any vertex $v \in V$ has four incident edges in $E(G^\diamond) \cup \{e_a,e_b\}$ and define its outer boundary 
\be \label{boundedge}
\delta V = \{e \in E(G^\diamond) \cup \{e_a,e_b\}, \text{only one end of $e$ belongs to V}\}.
\ee
For edge $e \in \delta V$, define the function $\eta_V(e)$ to be equal to $1$ if $e$ is pointing to a vertex of $V$ and $-1$ otherwise. 
Then, the following property holds \cite{DCST17}:

\be \label{parobssum}
\sum_{e \in \delta V} \eta_V(e) \hat{F}(e) =0.
\ee

\subsection{Graphs used in the proof}

\par In this paper, we work on $\Z^2$ but also on other infinite graphs. The half-plane $\{(x,y) \in \Z^2, y \ge 0\}$ is denoted by $\H$ (its boundary $\d \H$ is equal to $\{ (x,0) \in \Z^2\}$).

\par We can define the {\it{box}} of size $n$ in $\Z^2$ or in $\H$ as follows:
\benn
\L_n = \{(x,y) \in \Z^2 \text{(or $\H$)}: \max (|x|,|y|) \le n\},
\eenn
its boundary is naturally defined as
\benn
\d \L_n = \{(x,y) \in \Z^2 \text{(or $\H$)}: \max (|x|,|y|) = n\}.
\eenn

\par Free and wired boundary conditions are defined for $\H$ in the same way, as for $\Z^2$, as a limit of measures defined for boxes $\L_n \subset \H$ with free (resp. wired) boundary conditions on $\d \L_n$ and on $\d \H$. 
Also we work with $0\backslash 1$ boundary conditions, obtained by taking a limiting measure for boxes $\L_n$ with Dobrushin $((0,n),(0,0))$ boundary conditions (see Section \ref{paraf}).

\par The {\it{strip}} of height $n$ in $\Z^2$ (which can also be seen as a strip in $\H$) is defined as:
\benn \label{strip}
S_n = \{ (x,y) \in \Z^2: y\in [0,n] \}.
\eenn
Let us call the left and the right parts of its boundary the subsets defined as:
\begin{align}
\d^+S_n = \{ (x,y) \in \Z^2: y\in \{0,n\}, x \ge 0 \},\nonumber\\
\d^-S_n = \{ (x,y) \in \Z^2: y\in \{0,n\}, x \le 0 \}. \nonumber
\end{align}
The bottom left part of $\d S_n$ will be denoted by
\benn
\d_b^-S_n = \{ (x,0) \in \Z^2: x <0 \}.
\eenn

Free, wired or $0 \backslash 1$ boundary condition measures on $S_n$ are defined as a limit of measures on  rectangles 
$\{(x,y) \in \Z^2: x \in [-m,m], y \in [0,n]\}$ with free, wired or $((0,n),(0,0))$ boundary conditions respectively.

\par For some lemmas, we work on the universal cover of $\Z^2 \backslash F_0$ where $F_0$ is the face $[0,1]\times[-1,0]$. 
 The {\it{universal cover}} $\U$ is a graph defined as follows:
\begin{align}
V(\U) = \{& (x,y,z), x,y,z \in \Z\}= V(\Z^3),  \nonumber \\
E(\U) = \{&((x,y,z),(x,y+1,z)) \, \forall x,y,z \in \Z\} \nonumber \\
		\cup \,\{&((x,y,z),(x+1,y,z)) \, \forall x,y,z \in \Z \text{ such that } x \neq 0\}  \nonumber \\
		\cup \,\{& ((0,y,z),(1,y,z))\, \forall y,z \in \Z \text{ such that } y \ge 0 \}  \nonumber\\
		\cup \, \{& ((0,y,z),(1,y,z+1))\,\forall y,z \in \Z \text{ such that } y\le0\}. \nonumber
\end{align}
We also introduce the {\it{truncated universal cover}} $\U_k$ to be the subgraph of $\U$ with vertices of the type $(x,y,z)$ with $|z|\le k$. 
Let us define its boundary $\d \U_k$ as 
$\{(0,-y,-k), y \in \N\} \cup \{(0,-y,k), y \in \N\} $.

We can define the analogue of a box $\L_n$ in $\U_k$ as follows:
\begin{align}
\L_{n,k} &= \{(x,y,z) \in \U_k : \max (|x|,|y|) \le n\}, \nonumber \\
\d \L_{n,k} &= \{(x,y,z) \in \U_k: \max (|x|,|y|) = n\}. \nonumber
\end{align}

Free and wired boundary conditions are defined for $\U_k$ in the same way as for $\Z^2$. 
Let us also note that we can extend the notions of dual and medial graphs to $\U_k$.

 
\section{Proof of Theorem \ref{Th0} for $1 \le q \le 3$.} \label{Sketch3}

The proof is simpler when $q \le3$. We therefore begin by discussing this case.
Let us first notice that, by Zhang argument (see \cite{Gr06}), 
\be \label{Zhang}
p_c \ge p_{sd}.
\ee
We therefore have to show only the inequality $p_c \le p_{sd}$.


Let us take a set $S$ such that  $S \subset \L_n \subset \Z^2$ and $0 \in S$ and define its edge boundary $\Delta S = \{(x,y) \in E(\L_n): x \in S, y \not\in S\}$. Then, define the following auxiliary function:
\be \label{phi3}
\varphi_{p,q,n}(S)=\sum_{(x,y) \in\Delta S}{\phi_{S,p,q}^0[0 \arr x]}.
\ee 
The result is the consequence of the following two lemmas.

\begin {lemma} \label{q3dphidp}
Let $p \ge p_{sd}$, $q>1$ and $n \ge 1$. Then, for any $G$ such that $\L_n \subset G \subset \Z^2$, we have that
\be  \label{q3dphidpeq}
\frac{d}{dp} \phi_{G,p,q}^\xi[0\arr \d \L_n] \ge 
\frac{ c }{1-p} \left(\inf_{S:\, 0 \in S \subset \L_n}\varphi_{p,q,n}(S)\right) (1-\phi_{G,p,q}^\xi[0\arr \d \L_n])
\ee
where $c$ does not depend on $p, G$ or $\xi$.
\end{lemma}

\begin{proof} 
In this proof, we follow  \cite{DCT16, DCT216}.
The differential inequality \eqref{q3dphidpeq} is a consequence of \eqref{alltoclosed}, \eqref{difineq} and the characterisation of pivotal edges:
\be \nonumber
\frac{d}{dp} \phi_{G,p,q}^\xi[0\arr \d \L_n] \ge \frac{c}{1-p}\!\! \sum_{e\in E(G)} \! \phi_{G,p,q}^\xi[e \text{ is pivotal for } \{0\arr \d \L_n\} \text{ and } \{0 \nleftrightarrow \d \L_n\}].
\ee
Let us call $\mS$ the set $\{x \in \L_n: x\nleftrightarrow \d \L_n \}$. The event  $\{0 \nleftrightarrow \d \L_n\}$ can be rewritten as $\{0 \in \mathcal{S} \}$. Note that, when the event $\{0\arr \d \L_n\}$ does not occur, all the pivotal edges are closed and lie in $\Delta \mS $. Then, 
\begin{align}
\frac{d}{dp} \phi_{G,p,q}^\xi[0\arr \d \L_n] & \ge \frac{c}{1-p} \sum_{S:\, 0 \in S \subset \L_n}  \sum_{e\in E(G)} \! 
\phi_{G,p,q}^\xi[e \text{ is pivotal for } \{0\arr \d \L_n\}, \mS = S] \nonumber \\
& \ge \frac{c}{1-p} \sum_{S:\, 0 \in S \subset \L_n}  \sum_{x: (x,y) \in \Delta \mS} \! 
\phi_{G,p,q}^\xi[ \{0\underset{S}\arr x \}, \mS = S] \nonumber \\
& \ge \frac{c}{1-p} \sum_{S:\, 0 \in S \subset \L_n}  \! \varphi_{p,q,n}(S) \phi_{G,p,q}^\xi[ \mS = S] \nonumber 
\end{align}
where the last inequality follows from the Domain Markov property. 
Then, we can conclude that
\begin{align}
 \frac{d}{dp} \phi_{G,p,q}^\xi[0\arr \d \L_n] 
& \ge \frac{c}{1-p} \left( \inf_{S:\, 0 \in S \subset \L_n}  \varphi_{p,q,n}(S) \right) \phi_{G,p,q} ^\xi[ 0 \in \mS] \nonumber \\ 
&= \frac{c}{1-p} \left( \inf_{S:\, 0 \in S \subset \L_n}  \varphi_{p,q,n}(S) \right) (1-  \phi_{G,p,q}^\xi[0\arr \d \L_n]). \nonumber
\end{align}

\end{proof}

\begin{lemma} \label{q3phiC}
There exists $C>0$ such that for any $n \ge 1$ and $S$ such that $0 \in S \subset \L_n$,
\be 
\varphi_{p_{sd},q,n}(S) > C.
\ee
\end{lemma}

\begin{proof}
Consider a set $S$ such that $0 \in S \subset \L_n$.

We can define the graph $\L_n '$ as a box $\L_n$ with vertices $\{(0,k), 1 \le k \le n\}$ removed. Then, let us denote the connected component of $S$ that contains $0$  as $S_0$ and work with the domain $S_0' = S_0 \cap \L_n'$.
Let us call $(0)$ the boundary conditions that are free everywhere. 
These boundary conditions can be seen as Dobrushin boundary conditions with the wired arc collapsed to one point (i.e. $a=b=0$). 
This observation allows to define the exploration path $\gamma$ for a configuration in this domain.
Its beginning edge $e_a = \bigl( (0,\tfrac{1}{2}),(-\tfrac{1}{2},0)\bigr)$ is adjacent to the edge $e_b = \bigl( (\tfrac{1}{2},0),(0,\tfrac{1}{2})\bigr)$,
so $\gamma$ forms a loop around $0$, which bounds the open cluster in $S_0'$ that contains $0$.

\par Let us call $V$ the set of all vertices in $V((S_0')^\diamond)$ that have four incident edges in $E((S_0')^\diamond) \cup \{e_a,e_b\}$, then, $\delta V $ defined as in \eqref{boundedge} can be split into three parts: first, $\delta_0 = \{e_a,e_b\}$, second, the edges adjacent to the slit from the right and from the left (it can be written as $\delta_1 = \delta_1^+ \cup \delta_1^-$, where 
edges in $\delta_1^+$ are of the form $\bigl( (\tfrac{1}{2},k),(1,k+\tfrac{1}{2})\bigr)$ or $ \bigl( (\tfrac{1}{2},k+1),(1,k+\tfrac{1}{2})\bigr)$
and edges in $\delta_1^-$ are of the form $\bigl( (-\tfrac{1}{2},k),(-1,k+\tfrac{1}{2})\bigr)$ or $\bigl( (-\tfrac{1}{2},k+1),(-1,k+\tfrac{1}{2})\bigr)$),
third, other edges neighbouring the boundary of $(S_0')^\diamond$.

\par Let us also look at the domain $\overline S$ defined as the reflection of $S$ with respect to the $y$-axis. We can define $\overline {S'_0}$, $\overline V, \overline {\delta V},\overline {\delta_1^+}, \overline {\delta_1^-}$ and $\overline {\delta_2}$ in the same way as for $S$.

\par Then, \eqref{parobssum} can be written as
\benn
\sum_{e \in \delta_2} \eta_V(e) \hat{F}_{S'_0}(e)+
 \sum_{e \in \overline{\delta_2}} \eta_{\overline V}(e) \hat{F}_{\overline {S'_0}}(e)  +
 \sum_{e \in \delta_1 \cup \delta_0} \eta_V(e) \hat{F}_{S'_0}(e)
 +  \sum_{e \in \overline{\delta_2} \cup \delta_0} \eta_{\overline V}(e) \hat{F}_{\overline {S'_0}}(e)
 =0,
\eenn
or
\be \label{lemq3par}
 \bigl| \sum_{e \in \delta_1 \cup \delta_0} \eta_V(e) \hat{F}_{S'_0}(e)
 +  \sum_{e \in \overline{\delta_2} \cup \delta_0} \eta_{\overline V}(e) \hat{F}_{\overline {S'_0}}(e) \bigr| 
 \le \bigl|\sum_{e \in \delta_2} \eta_V(e) \hat{F}_{S'_0}(e)  \bigr| + 
\bigl| \sum_{e \in \overline{\delta_2}} \eta_{\overline V}(e) \hat{F}_{\overline {S'_0}}(e)  \bigr|.
\ee

The first term in the right part of  \eqref{lemq3par} is bounded as follows:
\begin{align}
\bigl|\sum_{e \in \delta_2} \eta_V(e) \hat{F}_{S'_0}(e) \bigr| 
&\le \sum_{e \in \delta_2} |\hat{F}_{S'_0}(e) | =
\sum_{e \in \delta_2} |\phi_{S'_0,p,q}^{(0)}(e^{i \sigma W_\gamma(e,e_b)} \one_{e \in \gamma})| \nonumber \\
&= \sum_{e \in \delta_2} \phi_{S'_0,p,q}^{(0)}(e \in \gamma)
= 2 \sum_{x \in \d S_0} \phi_{S',p,q}^0 (0 \arr x), \nonumber
\end{align}
because any boundary vertex corresponds to two edges from $\delta_2$ that do or do not belong to $\gamma$ simultaneously. The second term is bounded by the same value because of the symmetry between $S$ and $S'$:
\benn
\bigl|\sum_{e \in \delta_2} \eta_{\overline V}(e) \hat{F}_{\overline{S'_0}}(e) \bigr| 
\le 2 \sum_{x \in \d \overline S_0} \phi_{\overline{S_0'},p,q}^0 (0 \arr x)
\le 2 \sum_{x \in \d S_0} \phi_{S',p,q}^0 (0 \arr x).
\eenn
Together, this gives the following bound on the right part of  \eqref{lemq3par}:
\be \label{leftC3}
\bigl|\sum_{e \in \delta_2} \eta_V(e) \hat{F}_{S'_0}(e)  \bigr| \!+ \!\bigl|
  \sum_{e \in \overline{\delta_2}} \eta_{\overline V}(e) \hat{F}_{\overline {S_0'}}(e)  \bigr|
\le 4 \!\!\! \sum_{x \in \d S_0} \!\!\phi_{S',p,q}^0 (0 \arr x) 
\le 4 \!\!\! \sum_{x \in \d S_0}\!\! \phi_{S,p,q}^0 (0 \arr x).
\ee

\par 
Before writing the inequality for the right part of \eqref{lemq3par}, let us define the vertices of $S'_0$ and $\overline{S'_0}$ adjacent to the slit:
\begin{align}
\d^+_{slit} = \{ (1,k), 1 \le k \le n \}\cup S_0, \quad \overline{\d^+_{slit}}= \{ (1,k), 1 \le k \le n \}\cup \overline{S_0},\nonumber \\
\d^-_{slit} = \{  (1,k), 1 \le k \le n  \}\cup S_0, \quad \overline{\d^-_{slit}} = \{  (1,k), 1 \le k \le n  \}\cup \overline{S_0}.\nonumber
\end{align}

Each of these points corresponds to two edges of $\delta_1$ or $\overline{\delta_1}$.
Then, the sums in the left part of  \eqref{lemq3par} are written as follows:
\begin{align}
\sum_{e \in \delta_1 \cup \delta_0} \eta_V(e) \hat{F}_{S_0'}(e) = 
(e^{i \hs \frac{3 \pi}{2}} -1) + 
\sum_{x \in \d^+_{slit}}\phi^0_{S'_0, p_{sd},q}(x \arr 0) (e^{-i \hs \pi}-e^{-i\hs \frac{\pi}{2}}) \nonumber \\
+ \sum_{x \in \d^-_{slit}}\phi^0_{S'_0, p_{sd} ,q}(x \arr 0) (e^{2 \pi i \hs}-e^{i\hs \frac{5\pi}{2}}),
\nonumber \\ \nonumber
\sum_{e \in \overline{\delta_2} \cup \delta_0} \eta_{\overline V}(e) \hat{F}_{\overline {S'_0}}(e) = 
(e^{i \hs \frac{3 \pi}{2}} -1) + 
\sum_{x \in  \overline{\d^+_{slit}}} \phi^0_{\overline{S'_0}, p_{sd}, q}(x \arr 0) (e^{-i \hs \pi}-e^{-i\hs \frac{\pi}{2}}) \nonumber \\
+ \sum_{x \in \overline{\d^-_{slit}}} \phi^0_{\overline{S'_0}, p_{sd} ,q}(x \arr 0) (e^{2 \pi i \hs}-e^{i\hs \frac{5\pi}{2}}). \nonumber
\end{align}

Because of the symmetry of $S'_0$ and $\overline{S'_0}$, these equations can be summed to give
\begin{align}
\bigl|
\sum_{e \in \delta_1 \cup \delta_0} \eta_V(e) \hat{F}_{S'_0}(e) 
&+ \sum_{e \in \overline{\delta_2} \cup \delta_0} \eta_{\overline V}(e) \hat{F}_{\overline {S'_0}}(e)
\bigr|
= \bigl| 2(e^{i \hs \frac{3 \pi}{2}} -1) \nonumber \\ \nonumber
&+ \sum_{x \in \d^+_{slit}\cup \d^-_{slit}}\phi^0_{S'_0, p_{sd}, q}(x \arr 0) (e^{-\pi i \hs}-e^{-i\hs \frac{\pi}{2}} + e^{2 \pi i \hs}-e^{i\hs \frac{5\pi}{2}}) 
\bigr|
\\ \nonumber
&= \bigl| 2(e^{i \hs \frac{3 \pi}{2}} -1) + 
\sum_{x \in \d^+_{slit}\cup \d^-_{slit}} 2 \phi^0_{S'_0, p_{sd}, q}(x \arr 0) (e^{i \hs\frac{\pi}{2}} - e^{\pi i \hs}) \cos(\hs \tfrac{3\pi}{2})
\bigr|
\\ \nonumber
& =2 \bigl|  e^{\pi i \hs} -e^{i \hs \frac{\pi}{2}} \bigr|\cdot \bigl| \bigl(
(e^{i \hs \frac{\pi}{2}} +1+ e^{-i \hs \frac{\pi}{2}} ) - \!\!\! \!\!\!\sum_{x \in \d^+_{slit}\cup \d^-_{slit}} \phi^0_{S'_0, p_{sd}, q}(x \arr 0) \cos(\hs\tfrac{3\pi}{2})
\bigr) \bigr|
\\ \nonumber
& = 2 \bigl|e^{i \hs \frac{\pi}{2}}\bigr|\cdot\bigl| e^{i \hs \frac{\pi}{2}} -1 \bigr|\cdot\bigl|\bigl(
1+ 2 \cos( \hs \tfrac{\pi}{2}) -   \!\!\! \!\!\!\sum_{x \in \d^+_{slit}\cup \d^-_{slit}} \phi^0_{S'_0, p_{sd}, q}(x \arr 0) \cos(\hs\tfrac{3\pi}{2})
\bigr) \bigr|
\\ \nonumber
&= 2 
\bigl|e^{i \hs \frac{\pi}{2}} -1 \bigr| \cdot
\bigl| \bigl(
1 + \sqrt{q} + \!\!\! \!\!\!\sum_{x \in \d^+_{slit}\cup \d^-_{slit}}\phi^0_{S'_0, p_{sd}, q}(x \arr 0)  \tfrac{\sqrt{q}}{2}(3-q)
\bigr) \bigr|
\\ \label {rightC3}
&\ge 2 \bigl|e^{i \hs \frac{\pi}{2}} -1 \bigr| (1+ \sqrt{q}).
\end{align}
Together, \eqref{leftC3} and \eqref{rightC3} conclude the proof with $C = \frac{1}{2} |e^{i \hs \frac{\pi}{2}} -1| (1+ \sqrt{q})$.
\end{proof}



\begin{proof} [Proof of Theorem \ref{Th0} for $q \in {[}1,3{]}$]

Let us take $p' \ge p_{sd}$. By monotonicity, we can extend the result of Lemma \ref{q3phiC} to all values of $p$ in the interval $[p_{sd},p']$.
Thus, for all $p \in [p_{sd},p']$, \eqref{q3dphidpeq} takes the form
\benn
\frac{d \phi_{G,p,q}^\xi[0\arr \d \L_n]}{(1-\phi_{G,p,q}^\xi[0\arr \d \L_n])} \ge c\frac{d p}{1-p}
\eenn
or, written differently,
\be \label{prooft3din}
- d \log (1-\phi_{G,p,q}^\xi[0\arr \d \L_n]) \ge - d \log (1-p)^c.
\ee
We can integrate \eqref{prooft3din} on $[p_{sd},p']$ to obtain
\benn
\frac{1-\phi_{G,p_{sd},q}^\xi[0\arr \d \L_n]}{1-\phi_{G,p',q}^\xi[0\arr \d \L_n]} \ge \left( \frac{1-p_{sd}}{1-p'} \right)^c
\eenn
which gives
\be
\phi_{G,p',q}^\xi[0\arr \d \L_n] \ge 1 - \left(1-\phi_{G,p_{sd},q}^\xi[0\arr \d \L_n]\right) \left( \frac{1-p'}{1-p_{sd}}\right)^c
\ge 1- \left( \frac{1-p'}{1-p_{sd}}\right)^c
\ee
where $c$ does not depend on $G$ or $n$. We can send $G$ to $\Z^2$ and $n$ to infinity to finally obtain
\be
\phi_{\Z^2,p',q}^\xi[0\arr \infty] \ge 1- \left( \frac{1-p'}{1-p_{sd}}\right)^c > 0.
\ee
The probability to have an infinite cluster is therefore positive for any $p' > p_{sd}$, a fact which immediately implies that $p_c \le p_{sd}$. Together with $\eqref{Zhang}$ it gives Theorem \ref{Th0}.
\end{proof}


\section{Proof of Theorem \ref{Th0} for $3< q \le 4$.} \label{Sketch4}

The global strategy is almost the same in this case. We work in the strip $S_n$ rather than in the box $\L_n$.
Let us define  the event $ A^*_n = \{(0,0) \underset{S_n^*} \arr (0,n)\}$ on the dual lattice and call $P^*$ the left-most dual-open path connecting $(0,0)$ to $(0,n)$. We will also call $A_n$ the event complement to $A^*_n$.

We can the set $\mS= \{x \in S_n: x\nleftrightarrow \d^+ S_n \}$. The event $A^*_n$  is equal to the event $\{ \d^- S_n \subset \mS\}$.

We define the auxiliary function $\ovphi_{p,q,n}(S)$ as follows. 
Let us take a set $S$ such that $\d^-S_n \subset S \subset S_n$ and define $\Delta S = \{ (x,y) \in E(S_n): x \in S, y \not\in S\}$. Then
\be \label{phi4}
\ovphi_{p,q,n}(S)=\sum_{\{x,y\}\in\Delta S}{\phi_{S_n,p,q}^1[\d^-S_n \underset{S}\arr x \bigm| \mS = S]}.
\ee

\begin {lemma} \label{q4dphidp}
Let $p \ge p_{sd}$, $q>1$ and $n \ge 1$. Then, for any $G$ such that $S_n \subset G \subset \Z^2$, we have that
\be \label{q4dphidpeq}
\frac{d}{dp} \phi_{G,p,q}^1[A_n] \ge 
\frac{c}{1-p} \left(\inf_{S:\, \d^-S_n \subset S \subset S_n} \ovphi_{p,q,n}(S)\right) (1-\phi_{G,p,q}^1[A_n])
\ee
where $c$ does not depend on $p, G$ or $\xi$.
\end{lemma}

The analogue of Lemma \ref{q3phiC} for $\ovphi_{p,q,n}(S)$ is the key point of the proof and requires several additional statements. 
Firstly, we will show the following lemma:
\begin{lemma} \label{q4phiC1}
There exists a constant $c >0$ such that for any $n \ge 1$ and for any $S\in S_n$ with the properties that $\d^-S_n \subset S$ and $ \d^+ S_n \cap S = \emptyset$, we have
\be
\ovphi_{p,q,n}(S) \ge \sum_{x \in \d^-_b S_n} \phi^1_{S_n,p,q} [x \arr P^*].
\ee
\end{lemma}

In order to state the other lemma, we work on the truncated universal cover $\U_k$. The reason why we use the universal cover is the following. In the proof of Section \ref{Sketch3} (see \eqref{rightC3}), we used that $\cos(\tfrac{3 \pi}{2}\hat{\sigma}) < 1$ to show that the contribution of any slit has the same sign as the one of $e_a$ and $e_b$. In order to extend this property to $q>3$, one has to consider larger opening between strips. Then, the proof is very similar (Lemma \ref{q4phiCU}). 
One key observation will be that there is no infinite cluster in the universal case (Lemma \ref{Uinfcl}). 
Combining these two facts will lead (with some work, done in Lemma \ref{inasquare}) to an estimate on the plane (Lemma \ref{twocases})
This estimate will finally be used to show that the probability of a certain event decays very fast as $p$ moves away from $p_{sd}$,
a fact which is known to imply a bound in $p_c$ (Lemma \ref{lemcros}).

\par This lemma is the analogue of Lemma \ref{q3phiC} in $\U_k$:
\begin{lemma} \label{q4phiCU}
For any choice of $q \in (3,4]$, there exists a constant $C>0$ and $k \in \N$ such that for any $n \ge 1$ and any set $S$ such that $0 \in S \subset \L_{n,k}$, we have that
\be
\varphi_{p_{sd},q,n,k}(S) >C,
\ee
where $\varphi_{p,q,n,k}$ is defined in the same way as in \eqref{phi3} for sets included in $\L_{n,k}$.
\end{lemma}
\par We complement this lemma with the following result.
\begin{lemma} \label{Uinfcl}
For any $k \ge 0$, there is almost surely no infinite cluster in $\U_k$, i.e.
\be
\phi^0_{\U_k, p_{sd},q}(0 \arr \infty) = 0.
\ee
\end{lemma}
Combined together, these lemmas give the following technical estimate on $\Z^2$:
\begin{lemma} \label{inasquare}
For any $M>0$, there exists $R$ large enough such that for every $n > R$ and for any $\gamma$ connecting $\d \L_n$ and $\d \L _R$,
\benn
\sum_{(-i,-n), i \in [0,n]} \phi^\xi_{\L_n \backslash(\L_{R}\cup \gamma),p_{sd},q} [x \arr \L_{R}\cup \gamma]   \ge M,
\eenn
where $\xi$ denotes free boundary conditions on $\L_n$ and wired boundary conditions on $\L_{R}\cup \gamma$ (see Figure \ref{figlemsq}).
\end{lemma}

\begin{figure}[h!]
\centering
\includegraphics[width=0.5\textwidth]{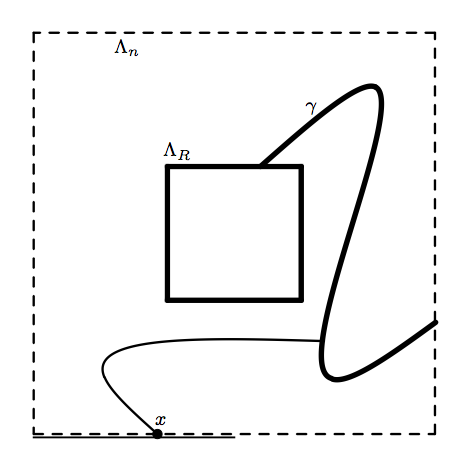}
\caption{Boundary point connected to $\L_R \cup \gamma$ (Lemma \ref{inasquare})}
\label{figlemsq}
\end{figure}
In this figure and the further pictures the free boundary conditions are represented by the dashed lines and the wired boundary conditions are represented by the bold lines.
\medskip

\par We use this lemma to obtain the following result.
\begin{lemma} \label{twocases}
Let us fix $\e>0$ and $R>0$. Then, for any $n$ large enough, one of the following statements should hold:
\begin{enumerate}
\item[Case 1] \label{case1}
\be  \label{q4phiCeq1}
\sum_{x \in \d^-_b S_n} \phi^1_{S_n,p_{sd},q} [x \arr 0] > R \log n.
\ee
\item[Case 2] \label{case2}
\be \label{case2eq}
\phi^0_{S_n,p_{sd}-\e,q} (0 \arr \d \L_n) < n^{-20}.
\ee
\end{enumerate}
\end{lemma}
For the second case, we can rewrite \eqref{case2eq} using the dual model on the dual lattice to get:
\be \label{T42eq}
\phi_{S_n,(p_{sd} +\delta)^* ,q}^{1^*}[A_n^*] \le \phi_{S_n,(p_{sd} +\delta)^* ,q}^{1^*}[0 \overset{*}\arr \d \L_n]  \le n^{-20}.
\ee

If the first case holds, then combined with Lemma \ref{q4phiC1}, it gives that for any $R > 0$, for $n$ large enough
\be  \label{q4phiCeq}
\ovphi_{p,q,n}(S) > R \log n
\ee
for any $S \in S_n$ with the properties $\d^-S_n \subset S$ and $ \d^+ S_n \cap S = \emptyset$.

The combination of this inequality with Lemma \ref{q4dphidp} implies the following proposition.
\begin{prop} \label{T4}
For $n$ such that \eqref{q4phiCeq1} holds and for any $\delta> 0$, we have that
\be \label{T41eq}
\phi_{S_n,(p_{sd} +\delta)^* ,q}^{1^*}[A_n^*] \le n^{-20}.
\ee
\end{prop}

\par The fact that $p_{sd} \ge p_c$ for $3 < q \le 4$ is the immediate consequence of \eqref{T42eq} and \eqref{T41eq} and the following lemma:
\begin{lemma}[\cite{DC13}] 
\label{lemcros} 
Suppose $p' < p_c$. Then, for infinitely many $n \in \N$, we have that
\be
\phi^{1^*}_{S_n,(p')^*,q} [A_n^*] > n^{-20}.
\ee
\end{lemma}

The rest of the paper is organised as follows.
In Section \ref{Proofthlem}, we show Lemma \ref{q4dphidp} and Theorem \ref{T4}.
Then in Section \ref{Prooflem1}, we use the parafermionic observable to prove Lemmas  \ref{q4phiC1} and \ref{q4phiCU}. 
Lemma \ref{Uinfcl} is proven in Section \ref{ProofUinfcl}.
Then, in Section \ref{square} we focus on Lemma \ref{inasquare}.
Lemma \ref{twocases} is the final step to conclude the proof and is shown in Section \ref{twoc}.


 
\section{Proofs of Lemma \ref{q4dphidp} and Proposition \ref{T4}.} \label{Proofthlem}
 These proofs use the same strategies as in the case $q \le 3$.

\begin{proof} [Proof of Lemma \ref{q4dphidp}]
(see also \cite{DCT16, DCT216}). 
Let us remind that the event complement to $A_n$ is equal to $\{ \d^- S_n \subset \mS\}$.
Then,
\begin{align} \nonumber
\frac{d}{dp} \phi_{S_n,p,q}^1[A_n] &\ge \frac{c}{1-p}\!\! \sum_{e\in E(S_n)} \! \phi_{S_n,p,q}^1[e \text{ is pivotal for } A_n, A_n \text{ does not occur}] \\
&\ge \frac{c}{1-p} \sum_{S:\,\d^-S_n \subset S \subset S_n} \sum_{(x,y) \in \Delta S} \phi_{S_n,p,q}^1 [\d^- S_n \underset{S} \arr x, \mS = S] \nonumber \\
&\ge \frac{c}{1-p} \sum_{S:\,\d^-S_n \subset S \subset S_n}  \ovphi_{p,q,n} (S) \phi_{S_n,p,q}^1 [ \mS = S] \nonumber \\
& \ge \frac{c}{1-p} \left( \inf_{S:\,\d^-S_n \subset S \subset S_n} \ovphi_{p,q,n} (S)  \right) (1 -\phi_{S_n,p,q}^1[A_n] ). \nonumber
\end{align}
\end{proof}

\begin{proof}[Proof of Proposition  \ref{T4}]
This proof uses the same method as for the proof given in Section \ref{Sketch3}. Fix $p' = p_{sd}+\delta$ for some $\delta>0$. By monotonicity, \eqref{q4phiCeq} holds for any $p \in [p_{sd}, p']$. 
This together with \eqref{q4dphidpeq} gives that
\benn
- d \left( \log(1-\phi_{G,p,q}^1[A_n])\right) \ge \frac{R}{1-p} \log n \, d p \ge R \log n \, dp.
\eenn
After integrating we obtain that 
\benn
\phi_{G,p_{sd}+\delta ,q}^1[A_n] \ge 1- n^{-R\delta}.
\eenn
Let us choose $R$ large enough to have $R \delta > 20$. Then, when $G$ goes to $S_n$, we obtain 
\benn
\phi_{S_n,p_{sd}+\delta ,q}^1[A_n] \ge 1- n^{-20}
\eenn
or, written in the dual model, 
\benn
\phi_{S_n,(p_{sd}+\delta)^* ,q}^{1^*}[A_n^*] \le n^{-20}.
\eenn
\end{proof}

\section {Proof of Lemmas \ref{q4phiC1} and \ref{q4phiCU}. }\label{Prooflem1}

\begin{proof} [Proof of Lemma \ref{q4phiCU}]

This proof uses the same strategy as in Lemma \ref{q3phiC}, but on $\U_k$ instead of $\Z^2$. Let us look at a set $S \in \L_{n,k}$ containing $0$ and denote $\overline{S} = \{(-x,y,-z), (x,y,z) \in S\}$ its reflection. 
We are interested in the connected component containing zero, denoted by $S_0$. 
Let us look at the set $\L_{n,k}' = \L_{n,k} \backslash \d\U_k$ and study the sets $S'_0 = \L_{n,k}' \cap S_0$ and $\overline{S'_0} = \L_{n,k}' \cap \overline{S_0}$.
\par For boundary conditions $(0)$ defined as before we can define the exploration path  both for $S'_0$ and $\overline{S'_0}$. Its initial and final edges are of the form $e_a = \bigl((0,-\tfrac{1}{2},-k)(\tfrac{1}{2},0,-k)\bigr)$ and $e_b = \bigl((-\tfrac{1}{2},0,k)(0,-\tfrac{1}{2},k)\bigr)$.
The sets $V, \delta V, \delta_0, \delta_1 = \delta_1^+ \cup \delta_1^-$ and $\delta_2$ (resp. $\overline{V}, \overline{\delta V}, \overline{\delta_1}=\overline{ \delta_1^+} \cup \overline{\delta_1^-}$ and $ \overline{\delta_2}$) are defined as in previous proof.

\par Then, \eqref{parobssum} enables us to write exactly the same equation as in \eqref{lemq3par}
\be \label{lemq4par}
 \bigl| \sum_{e \in \delta_1 \cup \delta_0} \eta_V(e) \hat{F}_{S'_0}(e)
 +  \sum_{e \in \overline{\delta_1} \cup \delta_0} \eta_{\overline V}(e) \hat{F}_{\overline {S'_0}}(e) \bigr| 
 \le \bigl|\sum_{e \in \delta_2} \eta_V(e) \hat{F}_{S_0'}(e)  \bigr| + 
\bigl| \sum_{e \in \overline{\delta_2}} \eta_{\overline V}(e) \hat{F}_{\overline {S_0'}}(e)  \bigr|.
\ee

As in  \eqref{leftC3}, the right-hand side of \eqref{lemq4par} is bounded as follows:
\be \label{leftC4}
\bigl|\sum_{e \in \delta_2} \eta_V(e) \hat{F}_{S_0'}(e)  \bigr| + \bigl|
  \sum_{e \in \overline{\delta_2}} \eta_{\overline V}(e) \hat{F}_{\overline {S_0'}}(e)  \bigr|
\le 4 \!\!\! \sum_{x \in \d S_0}\!\! \phi_{S,p_{sd},q}^0 (0 \arr x).
\ee

The left-hand side of \eqref{lemq4par} can be written as:
\begin{align}
\bigl|
\sum_{e \in \delta_1 \cup \delta_0} \eta_V(e) \hat{F}_{S'_0}(e) 
&+ \sum_{e \in \overline{\delta_1} \cup \delta_0} \eta_{\overline V}(e) \hat{F}_{\overline {S'_0}}(e)
\bigr|
= \bigl| 2(e^{i \hs (4\pi k\frac{3 \pi}{2})} -1) \nonumber \\ \nonumber
&+ \sum_{x \in \d^+_{slit}\cup \d^-_{slit}}\phi^0_{S'_0, p_{sd}, q}(x \arr 0) (e^{-\pi i \hs}-e^{-i\hs \frac{\pi}{2}} + e^{4k\pi i \hs }(e^{2 \pi i \hs}-e^{i\hs \frac{5\pi}{2}})) 
\bigr|
\\ \nonumber
&= \bigl| 2(e^{i \hs (4\pi k\frac{3 \pi}{2})} -1) 
\\ \nonumber
& + 
\sum_{x \in \d^+_{slit}\cup \d^-_{slit}} 2 \phi^0_{S'_0, p_{sd}, q}(x \arr 0) (e^{i \hs(2 \pi k +\frac{\pi}{2})} - e^{(2k+1)\pi i \hs}) \cos((4\pi k +3)\tfrac{\hs}{2})
\bigr|
\\ \nonumber
&= 2 \bigl| (e^{(2k+1)\pi i \hs}- e^{i \hs(2 \pi k +\frac{\pi}{2})}) 
\bigr|
\\ \nonumber
&\bigl| 
1+ \sum_{m = 1}^{4k+1} 2 \cos (\tfrac{\pi m}{2}\hs) - \sum_{x \in \d^+_{slit}\cup \d^-_{slit}} 2 \phi^0_{S'_0, p_{sd}, q}(x \arr 0) \cos((4\pi k +3)\tfrac{\hs}{2})
\bigr|
\\ \nonumber
&\ge 2 \bigl| (e^{i \hs \tfrac{\pi}{2}}- 1) 
\bigr|.
\end{align}

The last bound holds if we  pick an integer $k$ in  such a way that 
$\cos (\tfrac{\pi m}{2}\hs) \ge 0$ for any integer $m \in [0, 4k+1]$ and
$\cos ((4k+3) \tfrac{\pi }{2}\hs) \le 0$.
These inequalities give the constraint
\benn
k \in \biggl[\frac{\tfrac{2}{\hs}-3}{2}, \frac{\tfrac{2}{\hs}-1}{2}\biggr].
\eenn

The length of this interval is equal to one so such $k$ can always be found.



\end{proof}
\begin{proof} [Proof of Lemma \ref{q4phiC1}]
Look at the exploration path $\gamma$ from $(0,0)$ to $(n,0)$ in the strip $S_n$ with $0 \backslash 1$ boundary conditions. 
Take $S$ such that $\d^-S_n \subset S \subset S_n$ and its reflection $\overline {S}$ from the middle line of the strip. We call $V$ (correspondingly $\overline{V}$) the vertices of  $S^\diamond$ (correspondingly $\overline {S^{\diamond}}$), and $\delta_b^-, \delta_t^-, \delta_{P^*}$ and $\delta_{\overline{P^*}}$ the medial edges corresponding to the bottom and top left boundaries of $S_n$ and to the paths $P^*$ and $\overline{P^*}$.

The equation \eqref{parobssum} can be written as

\be \label{lemq4Spar}
 \bigl| \sum_{e \in \delta_b^- \cup \delta_t^-} \eta_V(e) \hat{F}_{S}(e)+ \eta_{\overline V}(e) \hat{F}_{\overline {S}}(e) \bigr| 
 \le \bigl|\sum_{e \in \delta_{P^*}} \eta_V(e) \hat{F}_{S}(e)  \bigr| + 
\bigl| \sum_{e \in \delta_{\overline{P^*}}} \eta_{\overline V}(e) \hat{F}_{\overline {S}}(e)  \bigr|.
\ee

The right part of the inequality is written as

\begin{align}
\bigl|\sum_{e \in \delta_{P^*}} \eta_V(e) \hat{F}_{S}(e)  \bigr| + 
\bigl| \sum_{e \in \delta_{\overline{P^*}}} \eta_{\overline V}(e) \hat{F}_{\overline {S}}(e)  \bigr|
&=
2\bigl|\sum_{e \in \delta_{P^*}} \eta_V(e) \hat{F}_{S}(e)  \bigr| 
\nonumber \\ \nonumber 
&\le
2\sum_{e \in \delta_{P^*}}  \phi^{0 \backslash 1}_{S_n, p_{sd}, q} (e \in \gamma) 
\nonumber \\ \nonumber 
&= 4 \sum_{e \in P^*} \phi^{0 \backslash 1}_{S_n, p_{sd}, q} (e \arr \d^-S_n) 
\nonumber \\ \nonumber 
&= 8 \sum_{e \in P^*}\phi^{0 \backslash 1}_{S_n, p_{sd}, q}  (e \arr \d^-_bS_n).
\end{align}
due to the symmetry of $S$ and $\overline{S}$. Also, because of this symmetry, the left-hand side of \eqref{lemq4Spar} is bounded by

\begin{align}
\bigl| 
(e^{i\tfrac{\pi}{2} \hs}-1 +e^{-i\tfrac{\pi}{2} \hs} -e^{-i\pi \hs}) 
\!\!\! \sum_{x \in \d^-S_n} \!\!\phi^{0 \backslash 1}_{S_n, p_{sd}, q} (x \overset{*} \arr P^*)
\bigr|
& =
\tfrac{\sqrt{q}}{2} \bigl| 1 - e^{-i\tfrac{\pi}{2} \hs}\bigr| 
\!\!\! \sum_{x \in \d^-S_n}\!\! \phi^{0 \backslash 1}_{S_n, p_{sd}, q}(x \overset{*} \arr P^*)
\nonumber \\ \nonumber &=
\sqrt{q}\bigl| 1 - e^{-i\tfrac{\pi}{2} \hs}\bigr| 
\!\!\!\sum_{x \in \d^-_bS_n} \!\! \phi^{0 \backslash 1}_{S_n, p_{sd}, q} (x \overset{*} \arr P^*).
\end{align}
The combination of these two bounds finishes the proof.
\end{proof}

 \section{Proof of Lemma \ref{Uinfcl}} \label{ProofUinfcl}

To prove this lemma we introduce new definitions and prove one intermediate lemma. 
For every integer $r$, call $\ell_r$ the axis in $\U$ obtained by the rotation of $\ell_0~=~\{(0,y,0)~:~y~\ge~0\}$ by the angle $\tfrac{r \pi}{2}$. 
Let us denote $\mathrm{sect}_{i, j}$ the part of $\U$ between $\ell_i$ and $\ell_{j}$ 
(in particular $\U_k~=~\mathrm{sect}_{-2-4k,4k+2}$). 

\par Let us call the domain $\Omega \subset \U$ {\it{symmetric}} if it is invariant under the reflection with respect to $\ell_0$, i.e, if $(x,y,z) \in \Omega$ implies that $(x,-y,-z) \in \Omega$. 
Let us also call $\Omega$ {\it{simple}} if $0 \in \Omega$ and if for any sector $\mathrm{sect}_{i, i+1}$ the domain $\mathrm{sect}_{i, i+1} \backslash \Omega$ is connected.

\begin{lemma} \label{prop1}
For any $k >0$ and any
simple symmetric domain $\Omega$ such that $ \Omega \subset \mathrm{sect}_{-k,k}$ 
the following is true:
\be
\phi^{0,0\backslash1}_{\mathrm{sect}_{-k,k} \backslash \Omega, p_{sd},q} \bigl(\ell_1 \overset{*}{\underset{\mathrm{sect}_{-1,1}}\arr} (\d \Omega \cap \mathrm{sect}_{-1, 0}) \bigr) \ge \frac{1}{2}.
\ee
where $0,0\backslash1$ boundary conditions denote free boundary conditions at infinity, $\ell_k$ and $\ell_{-k}$, free boundary conditions on $\mathrm{sect}_{-k,0}\cap \d \Omega$ and wired boundary conditions on $\mathrm{sect}_{0,k}\cap \d \Omega$.
\end{lemma}

\begin{figure}[h!]
\centering
\includegraphics[width=\textwidth]{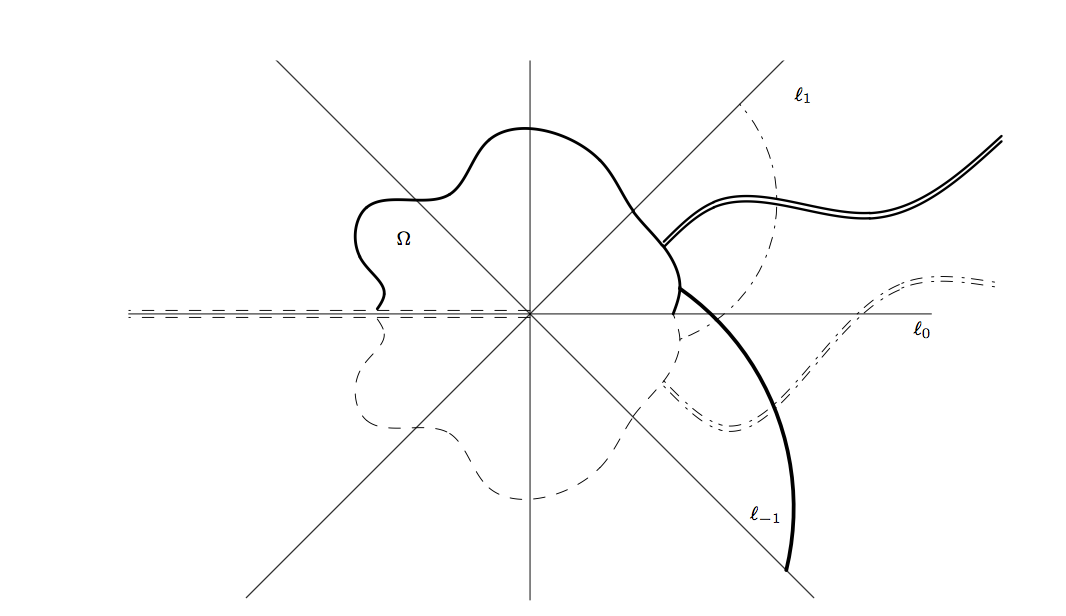}
\caption{A domain $\Omega$ with corresponding boundary conditions, events $A_1$ and $A_2$, and $B_1$ and $B_2$ (doubled lines). Dual paths are represented by dash-dotted lines.}
\label{figprop1}
\end{figure}

\begin{proof}
The events 
$$A_1 = \{\ell_1 \overset{*}{\underset{\mathrm{sect}_{-1,1}\backslash \Omega}\arr} (\d \Omega \cap \mathrm{sect}_{-1,0}) \}$$
and 
$$A_2 = \{\ell_{-1} {\underset{\mathrm{sect}_{-1,1}\backslash \Omega}\arr} (\d \Omega \cap \mathrm{sect}_{0, 1}) \}$$
are disjoint. 
Let us denote $1,0\backslash1$ the boundary conditions which are wired at infinity, $\ell_k$ and $\ell_{-k}$, free at $\mathrm{sect}_{-k,0}\cap \d \Omega$ and wired at $\mathrm{sect}_{0,k}\cap \d \Omega$.
Then, by duality and by comparison between free and wired boundary conditions (which favour primal paths to appear), 
$$\phi^{0,0\backslash1}_{\mathrm{sect}_{-k,k} \backslash \Omega, p_{sd},q}(A_1) = 
\phi^{1,0\backslash1}_{\mathrm{sect}_{-k,k} \backslash \Omega, p_{sd},q}(A_2) \ge
\phi^{0,0\backslash1}_{\mathrm{sect}_{-k,k} \backslash \Omega, p_{sd},q}(A_2).
$$
Let us look at the events
$$B_1 = \{ (\d \Omega \cap \mathrm{sect}_{0, 1}){\underset{\mathrm{sect}_{-1,1}\backslash \Omega}\arr} \infty \}$$
and
$$B_2 = \{ (\d \Omega \cap \mathrm{sect}_{-1,0})\overset{*}{\underset{\mathrm{sect}_{-1,1}\backslash \Omega}\arr} \infty \}.$$
The realisations of $A_2$ and $B_1$ are the only paths blocking the event $A_1$ (see Figure \ref{figprop1}). 

Thus, $A_1$, $A_2$ and $B_1 \cap B_2$ are disjoint, and moreover, $A_1 \cup A_2 \cup (B_1 \cap B_2)$ is equal to the full probability space:
\be \label{fullspace}
\phi^{0,0\backslash1}_{\mathrm{sect}_{-k,k} \backslash \Omega, p_{sd},q}(A_1) +
\phi^{0,0\backslash1}_{\mathrm{sect}_{-k,k} \backslash \Omega, p_{sd},q}(A_2) +
\phi^{0,0\backslash1}_{\mathrm{sect}_{-k,k} \backslash \Omega, p_{sd},q}(B_1 \cap B_2)  = 1
\ee

Let us bound the probability of $ B_1 \cap B_2$. We can compare it to the event 
\begin{align} 
\tilde{B} = \{ 
0 \overset{*}{\underset{\mathrm{sect}_{-1,1}}\arr} \infty, 
0 {\underset{\mathrm{sect}_{-1,1}}\arr} \infty
&\text{ and  the dual cluster comes before the primal one} 
\nonumber\\ \nonumber 
&\text{ on the way from $\ell_{-1}$ to $\ell_1$}\}.
\end{align}
We can open all edges of $\Omega \cap \mathrm{sect}_{0,1}$ and close all edges of $\Omega \cap \mathrm{sect}_{-1,0}$, this connects the primal cluster from the event $B_1$ and the dual cluster from the event $B_2$ to zero. 
Then by finite energy property there exists a positive constant $c$ depending only on $\Omega$, such that
\begin{align}
c \phi^{0,0\backslash1}_{\mathrm{sect}_{-k,k} \backslash \Omega, p_{sd},q}(B_1 \cap B_2) 
&\le  \phi^{0}_{\mathrm{sect}_{-k,k}, p_{sd},q}(B_1 \cap B_2 \cap 
\{(\Omega \cap \mathrm{sect}_{-1,0}) \text{ closed, }(\Omega \cap \mathrm{sect}_{-1,0}) \text{ open}\})
\nonumber \\ \nonumber
& \le  
 \phi^{0}_{\mathrm{sect}_{-k,k}, p_{sd},q}(\tilde{B}).
\end{align}

From now on the proof will require that $k \ge 3$, but it can be easily modified for $k \in \{1,2\}$.

\par Let us look at the probability of the event 
$$C = \{0 \overset{*}{\underset{\mathrm{sect}_{1,3}}\arr} \infty\} $$ 
conditioned on $\tilde{B}$, 
and compare it to the event $\tilde{B}$ itself. 
The existence of the primal cluster from $0$ to infinity in $\mathrm{sect}_{-1,1}$ has less influence on $C$, 
than a primal cluster  from $0$ to infinity in $\mathrm{sect}_{1,-3}$ (closer to $\ell_{1}$, than the dual cluster).
The free boundary conditions at $\ell_{k}$ are two $\tfrac{\pi}{2}$-turns closer to the dual cluster from $C$ than the free boundary conditions at $\ell_{-k}$ for the dual cluster of $\tilde{B}$.
The comparison between boundary conditions concludes that
$$
\phi^0_{\mathrm{sect}_{-k,k}, p_{sd},q} (C | \tilde{B}) \ge
\phi^0_{\mathrm{sect}_{-k,k}, p_{sd},q} (0 \overset{*}{\arr} \infty \text{ in } \tilde{B}|0 \arr \infty \text{ in } \tilde{B}) \ge
\phi^0_{\mathrm{sect}_{-k,k}, p_{sd},q} (\tilde{B}),
$$
and, by comparison between boundary conditions,
$$
\phi^0_{\mathrm{sect}_{-3,1}, p_{sd},q} (C \cap \tilde{B}) \ge
\phi^0_{\mathrm{sect}_{-k,k}, p_{sd},q} (C \cap \tilde{B}) \ge
\bigr(\phi^0_{\mathrm{sect}_{-k,k}, p_{sd},q} (\tilde{B})\bigl)^2.
$$
Let us look at the existence of the primal infinite cluster in $C \cap \tilde{B}$ conditioned on the existence of two separated infinite dual clusters (let us call this event $(0 \overset{*}{\arr} \infty)^2$). The probability of this event will not change if the boundaries $\ell_{-1}$ and $\ell_3$ are glued together to obtain $\Z^2$. The probability for primal infinite cluster to exist increases if we remove two dual clusters. Thus,
\begin{align}
\bigr (c \phi^{0,0\backslash1}_{\mathrm{sect}_{-k,k} \backslash \Omega, p_{sd},q}(B_1 \cap B_2) \bigl)^2
&\le \phi^0_{\mathrm{sect}_{-3,1}, p_{sd},q} (C \cap \tilde{B}) \nonumber \\
&\le \phi^0_{\mathrm{sect}_{-3,1}, p_{sd},q} (0 {\arr} \infty \text{ in } C \cap \tilde{B}|(0 \overset{*}{\arr} \infty)^2 \text{ in } C \cap \tilde{B})  \nonumber \\ \nonumber
&\le \phi^0_{\Z^2, p_{sd},q} (0 {\arr} \infty \text{ in } C \cap \tilde{B}|(0 \overset{*}{\arr} \infty)^2 \text{ in } C \cap \tilde{B})\\ \nonumber
& \le \phi^0_{\Z^2, p_{sd},q} (0 \arr \infty),
\end{align}
and the last probability is equal to zero because of Zhang's argument. 
Combined with \eqref{fullspace}, this implies the result.

\end{proof}

\begin{proof}[Proof of Lemma \ref{Uinfcl}]

\par We are going to prove that, with positive probability, there exists a dual path in $\U_k \backslash \L_{n,k}$ disconnecting $0$ from infinity and that this probability does not depend on $n$. This fact implies the statement of the proposition. 

Let us look at the event $\{ \ell_{-4k-2} \overset{*}{\underset{\U_k \backslash \L_{n,k}}\arr} \ell_{-4k-1}\}$ not in $\U_k = \mathrm{sect}_{-4k-2, 4k+2}$, but in a bigger domain $\mathrm{sect}_{-12k-6, 4k+2}$.
The domain $\Omega = \L_{n,3k+1} \cap \mathrm{sect}_{-12k-6,4k+2}$ is simple and symmetric with respect to the symmetry line $\ell_{-4k-2}$.
Putting free boundary conditions on $\ell_{-12k-6}$ and on $\d \Omega \cap \mathrm{sect}_{-12k-6,-4k-2} $ instead of $\ell_{-4k-2}$ decreases the probability of dual path to appear. Then, by comparison between boundary conditions



\begin{align}
\phi^{0}_{\U_k,p_{sd},q} (\ell_{-4k-2} & \overset{*}{\underset{\U_k \backslash \L_{n,k}}\arr} \ell_{-4k-1}) \nonumber
\\ & \ge \phi^{0}_{\mathrm{sect}_{-12k-6,4k+2},p_{sd},q} 
(\ell_{-4k-2} \!\!\!\!\! \overset{*}{\underset{\mathrm{sect}_{-12k-6,4k+2} \backslash \Omega}\arr} \!\!\!\!\! \ell_{-4k-1}| \Omega \cap \mathrm{sect}_{-12k-6,-4k-2} \text{ closed})  
\nonumber \\ \nonumber
& \ge \phi^{0,0\backslash1}_{\mathrm{sect}_{-12k-6,4k+2} \backslash \Omega,p_{sd},q} 
(\ell_{-4k-2} \overset{*}{\arr} \ell_{-4k-1})  
\ge \tfrac{1}{2},
\end{align}
where the last bound is a direct consequence of Lemma \ref{prop1}.


%
%
%
%

\begin{figure}[!h]
\centering
\includegraphics[width=0.85\textwidth]{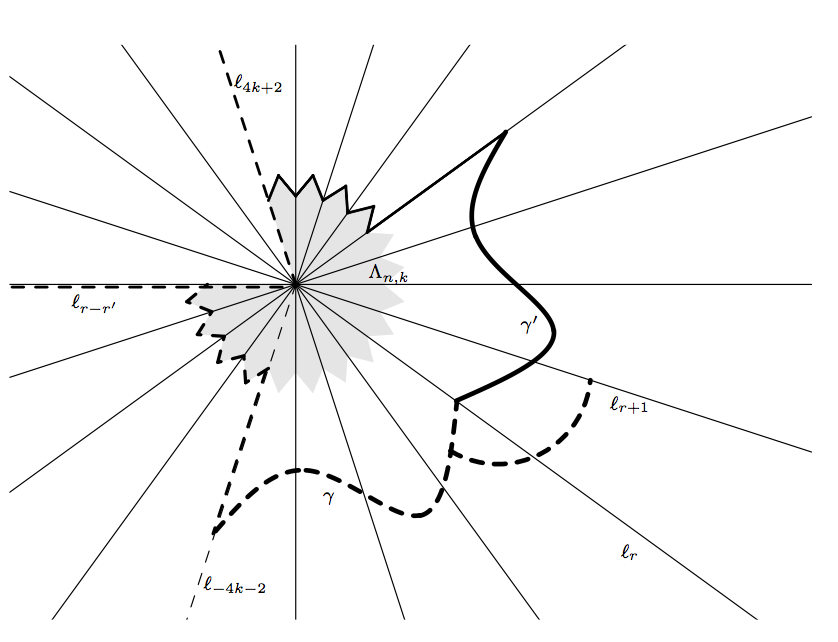}

\caption{Construction of the dual connection between $\gamma$ and $\ell_{r+1}$}
\label{figUinfcl}
\end{figure}

Suppose now that for some integer $r \in [-4k-1, 4k+1]$ lines $\ell_{-4k-2}$ and $\ell_r$ are already connected in $\U_k \backslash \L_{n,k}$ by a dual path $\gamma$ (see Figure \ref{figUinfcl}).
Consider $\ell_r$ as a symmetry line and reflect $\gamma$ with respect to it (let us call the result $\gamma'$). 
The domain $\Omega'$ defined as the area of $\mathrm{sect}_{-4k-2,2r+4k+2}$ bounded by $\g \cup \g'$ is a simple symmetric domain.
We are going to work in a sector $\mathrm{sect}_{r-r',r+r'}$, where $r' = 4k+2+|r|$. 
Then $\Omega = (\L_{n,3k+1} \cap \mathrm{sect}_{r-r',r+r'}) \cup \Omega'$ is also a simple symmetric domain and, using the same strategy as before, we obtain

\begin{align}
\phi^{0}_{\U_k,p_{sd},q} (\ell_{r} \overset{*}{\underset{\U_k \backslash \L_{n,k}}\arr} \ell_{r+1} \bigm| \gamma ) 
& \ge \phi^{0}_{\mathrm{sect}_{r-r',r+r'},p_{sd},q} 
(\ell_{r} \overset{*}{\underset{\mathrm{sect}_{r-r',r+r'} \backslash \Omega}\arr} \ell_{r+1}\bigm|
\Omega \cap \mathrm{sect}_{r-r',r} \text{ closed})  
\nonumber \\ \nonumber
& \ge \phi^{0,0\backslash1}_{\mathrm{sect}_{r-r',r+r'} \backslash \Omega,p_{sd},q} 
(\ell_{r} \overset{*}{\arr} \ell_{r+1})  
\ge \tfrac{1}{2}.
\end{align}

This implies the result since, by iterative conditioning, we find 
$$
\phi^{0}_{\U_k,p_{sd},q} (\ell_{-4k-2} \overset{*}{\underset{U_k \backslash \L_{n,k}}\arr} \ell_{4k+2}) \ge\bigr( \tfrac{1}{2}\bigl)^{8k+4}.
$$

\end{proof}

\section{Proof of Lemma \ref{inasquare}} \label{square}

\begin{proof}
Fix $k$  and $C$ as in Lemma \ref{q4phiC1} so that:
\benn
\sum_{x \in \d \L_{n,k}} \phi^0_{\L_{n,k},p_{sd},q} (0 \arr x) \ge C.
\eenn

Divide the boundary of $\L_{n,k}$ into $8(2k+1)$ pieces by splitting each side of each boundary layers into two halfs at the midpoint. Then, there exists at least one piece $I_{n,k}$ such that 
\be \label{1sq}
\sum_{x \in I_{n,k}} \phi^0_{\L_{n,k},p_{sd},q} (0 \arr x) \ge \frac{C}{17 k}.
\ee
By the finite energy property, there exists a constant $c' = c'(q,k,\xi) \ge 1$ such that for any $j \in \Z \cap [-k,k]$ and for any $x \in \U_k$ the following holds:
\be \label{fepsq}
\frac{1}{c'} \ge\frac{\phi^\xi_{\U_k,p,q}[x \arr (0,0,0)]}{\phi^\xi_{\U_k,p,q}[x \arr (0,0,j)]} \ge c'.
\ee
Combining \eqref{1sq} and \eqref{fepsq}, we obtain that for some positive constant $c$, independent of $n$,
\benn
\sum_{x \in  I_{n,k}} \phi^0_{\L_{n,k},p_{sd},q} [(0,0,j) \arr x] \ge c,
\eenn
where $j$ is the height of $I_{n,k}$. 
Using a bigger domain with $2k$ layers  to centre the $j$-th layer and comparing boundary conditions as in \eqref{boundcond1}, we conclude that
\benn
\sum_{x \in  I_{n}} \phi^0_{\L_{n,2k},p_{sd},q} [0 \arr x] \ge c,
\eenn
where $I_n$ is the projection of $I_{n,k}$ onto the layer of height $0$. 
\par Let us open all edges in a smaller box $\L_{R, 2k}$ for some $R \in (0, N)$. Then, we can write the following inequality
\begin{align}
\sum_{x \in  I_{n}} \phi^0_{\L_{n,2k},p_{sd},q} [x \arr \L_{R,2k}| \L_{R, 2k} \text{ is all open}] 
&\ge \frac{\sum_{x \in  I_{n}} \phi^0_{\L_{n,2k},p_{sd},q} [0 \arr x] }{ \phi^0_{\L_{n,2k},p_{sd},q} [ 0 \arr \d \L_{R,2k}]} \nonumber \\
& \ge \frac{c }{{ \phi^0_{\L_{n,2k},p_{sd},q} [ 0 \arr \d \L_{R,2k}]}}. \nonumber
\end{align}
Also, we can choose any $\gamma$ from $\d \L_{n}$ to $\d \L_R$, project it on all layers of $\U_k$ as $\tilde \gamma$ and make $\tilde \gamma$ open. Then, we obtain 
\begin{align}
\sum_{x \in  I_{n}} \phi^\xi_{\L_{n,2k}\backslash(\L_{R, 2k}( \tilde\gamma)),p_{sd},q} [x \arr \L_{R,2k}( \tilde \gamma)] 
&= \sum_{x \in  I_{n}} \phi^0_{\L_{2n,k},p_{sd},q} [x \arr \L_{R,2k}| \L_{R, 2k} \cup \tilde\gamma \text{ are open}] \nonumber \\
& \ge \frac{c }{{ \phi^0_{\L_{n,2k},p_{sd},q} [ 0 \arr \d \L_{R,2k}]}}. \nonumber
\end{align}
where $\L_{R,2k}( \tilde \gamma)$ is a union of $\L_{R,2k}$ and $\tilde \gamma)$, and $\xi$ denotes the boundary conditions wired on $ \d \L_{R, 2k}$ and $\tilde\gamma$ and free on $\L_{n,2k}$. Notice that $\{x: x \underset{\L_{n,2k}\backslash(\L_{R, 2k}\cup \tilde\gamma)}\arr  I_{n}\}$ can be projected on $\Z^2$ without multiple projections on one point (if not, an open path between two points with the same projection would cross $\tilde \gamma$). Thus, we can conclude that
\be \label{inasquare1}
\sum_{x \in  I_{n}} \phi^\xi_{\L_n \backslash(\L_{R}\cup \gamma),p_{sd},q} [x \arr \L_{R}\cup \gamma]   \ge \frac{c }{{ \phi^0_{\L_{n,2k},p_{sd},q} [ 0 \arr \d \L_{R,2k}]}}. 
\ee
By Lemma \ref{Uinfcl}, for any $\e > 0$, one may choose $R$ large enough that for $n$ large enough 
\benn
\phi^0_{\L_{n,2k},p_{sd},q} [ 0 \arr \d \L_{R,2k}] < \e.
\eenn
Together with \eqref{inasquare1}, this gives the result.

\end{proof}

\section{Proof of Lemma \ref{twocases}} \label{twoc}
Let us call $p_{m,n} = p_{m,n}(C)$ the probability that the box $[-m,m] \times [0, Cm] \in \H$ with wired boundary conditions on $\partial S_n$ is crossed from top to bottom, i.e.
\be
p_{m,n}  = \phi^1_{S_n, p_{sd},q} ([-m,m] \times \{0\} \underset{[-m,m] \times [0, Cm]} \arr [-m,m] \times \{ Cm\}).
\ee

Let us fix $C \in \N$ and take $k, m, n \in \N$ such that $ 2 C k m < n$. For $i \ge 0$, define the domains $L^b_i$ and $L_i^t$ as follows:
\beenn
&&L_0^b = {[-m,m] \times [0, Cm]}, \\
&&L_0^t = {[-m,m] \times [n -Cm,n]}, \\
&& L_i ^b  = L_{i-1}^b \cup \bigl( L_{i-1}^b + (-2m, 0) \bigr) \cup  \bigl( {[-m,m] \times [0, Cm]} + (-im, iCm)\bigr), \\
&& L_i ^t = L_{i-1}^t \cup \bigl( L_{i-1}^t + (-2m, 0) \bigr) \cup  \bigl( {[-m,m] \times [n -Cm,n]} + (-im, -iCm)\bigr).
\eeenn
For $i \le k$, the domains $L_i^t$ and $L_i^b$ stay in the strip $S_n$ and do not intersect.
We are going to study the events 
\beenn 
&& A_{m,n,i}^b = \bigl \{   [0,m]\times \{0\} \underset{L_i^b} \arr \{-(i+1)m\}\times  \Z \bigr \},\\
&& A_{m,n,i}^t = \bigl \{   [0,m]\times \{n\} \underset{L_i^t} \arr \{-(i+1)m\}\times  \Z \bigr \}.
\eeenn 

\begin{lemma}
For every $k,m,n \in \N$  such that $ 2 C k m < n$, we have that
\be
\phi^{0\backslash 1}_{S_n, p_{sd},q} (A_{m,n,k}^b \cap A_{m,n,k}^t) \ge (\tfrac{1}{2} - p_{m,  n- 2Ckm})^{2(k+1)}.
\ee
\end{lemma}

\begin{proof}
%
%

We prove the estimate by induction. The event $A_{m,n,0}^b$ is rewritten as follows:
$$ \bigl\{ ([0,m] \times \{0\} \underset{[-m,m] \times [0, Cm]} \arr \{-m\} \times [0, Cm]) \bigr\} $$.

Under $0 \backslash 1$ boundary conditions, the events
$$ \bigl\{ [0,m] \times \{0\} \underset{[-m,m] \times [0, Cm]} \arr ([-m,m] \times \{ Cm\} )\cup (\{-m\} \times [0, Cm] \bigr\} $$
and
$$\bigl\{[-m,0] \times \{0\} \overset{*} {\underset{[-m,m] \times [0, Cm]} \arr } ([-m,m] \times \{ Cm\}) \cup (\{m\} \times [0, Cm]) \bigr\} $$
have the same probability. Since by duality and symmetry at least one of them should occur in any configuration, we find that
$$\phi^{0\backslash 1}_{S_n, p_{sd},q}( [0,m] \times \{0\} \underset{[-m,m] \times [0, Cm]} \arr ([-m,m] \times \{ Cm\} )\cup (\{-m\} \times [0, Cm]) )\ge \frac{1}{2}.$$
Adding wired boundary conditions increase the probability to have a vertical crossing of the box so
$$ \phi^{ 0 \backslash 1}_{S_n, p_{sd},q} ([-m,m] \times \{0\} \underset{[-m,m] \times [0, Cm]} \arr [-m,m] \times \{ Cm\}) \le p_{m,n}$$
which implies that
$$\phi^{0\backslash 1}_{S_n, p_{sd},q}(A^b_{m,n,0}) = \phi^{0\backslash 1}_{S_n, p_{sd},q} ([0,m] \times \{0\} \underset{[-m,m] \times [0, Cm]}\arr \{-m\} \times [0, Cm]) \ge \frac{1}{2} - p_{m,n}.$$
The same estimation is true also for $A^t_{m,n,0}$. Then, the FKG inequality gives that 
$$
\phi^{0\backslash 1}_{S_n, p_{sd},q}(A^b_{m,n,0} \cap A^t_{m,n,0}) \ge \bigl( \frac{1}{2} - p_{m,n} \bigr)^2.
$$

\begin{figure}[!h]
\centering
\includegraphics[width=0.75\textwidth]{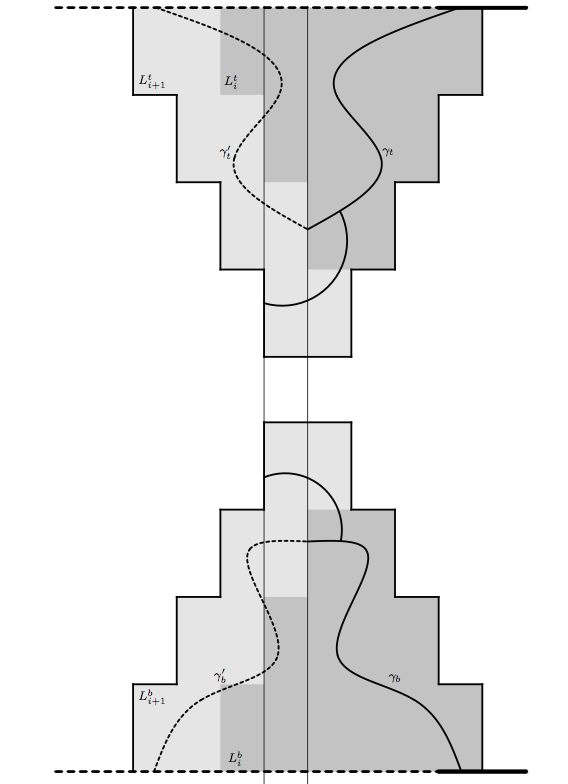}

\caption{Events $A^b_{m,n,i+1}$ and $A^t_{m,n,i+1}$ with corresponding boundary conditions.}
\label{figtopbot}
\end{figure}

Suppose now that for some $i$,
$$\phi^{0 \backslash 1}_{S_n, p_{sd},q}(A^b_{m,n,i} \cap A^t_{m,n,i}) \ge (\tfrac{1}{2} - p_{m,n - 2Cim})^{2(i+1)}.$$
Let us first notice that $p_{m,n}$ decreases with respect to the second argument because of the Domain Markov property (narrowing the strip by adding wired boundary conditions increases the probability to have a crossing in a box inside it). Thus, we can write that
\be \label{hugoadd}
\phi^{0 \backslash 1}_{S_n, p_{sd},q}(A^b_{m,n,i} \cap A^t_{m,n,i}) 
\ge (\tfrac{1}{2} - p_{m,n - 2C(i+1)m})^{2(i+1)}.
\ee

\par Let us look at the event $A^b_{m,n,i+1}$ conditioned on the events $A^b_{m,n,i}$ and $A^t_{m,n,i}$ (see Figure \ref{figtopbot}). Let $\gamma_b$ be the uppermost path from  $[0,m]\times \{0\}$ to $\{-(i+1)m\}\times  \Z$ satisfying $A^b_{m,n,i}$ and $\gamma_t$ be the lowermost path from  $[0,m]\times \{n\}$ to $\{-(i+1)m\}\times  \Z$ satisfying $A^t_{m,n,i}$.  Let $\gamma_b'$ and $\gamma_t'$ be the reflections of $\gamma_b$ and $\gamma_t$ with respect to  the line $\{-(i+1)m\}\times  \Z$.
Note that $\{-(i+1)m\}\times  \Z$ is an axis of symmetry for $L^b_{i+1}$ and $L^t_{i+1}$.

Set the boundary conditions to be free on $\g'_b$ and $\g'_t$ (this can only decrease the probability of $A^b_{m,n, i+1}$). Then, by the same reasons as for $A_{m,n,0}$, the probability to have an open path from $([0,n]\times \{0\}) \cup \g$ to the top or to the left boundary of $L^b_{i+1}$ is bigger than $\tfrac{1}{2}$.

The probability for this path to hit the top part of $\d L^b_{i+1}$ (i.e. $[-(i+2)m,-im] \times \{(i+2)Cm\} $) is smaller than the probability for $  [-m,m] \times [0, Cm] + (-(i+1)m, (i+1)Cm)$ to be crossed from top to bottom, which can be bounded above using the comparison between boundary conditions and the Domain Markov property. 
If we restrict our strip to $\Z \times [(i+1)Cm, n- (i+1)Cm]$ (all the paths will be left outside the strip) and put wired boundary conditions on its boundary, the probability $  [-(i+2)m,-im] \times [(i+1)Cm, (i+2)Cm]$ to have a vertical crossing is equal to $p_{m,n-2C(i+1)m}$. The initial domain is bigger and has smaller boundary conditions, so the probability of this event is smaller in this context.
This leads to
\benn
\phi^{0\backslash 1}_{S_n, p_{sd},q}(A^b_{m,n,i+1} | A^b_{m,n,i} \cap A^t_{m,n,i}) \ge \tfrac{1}{2} - p_{m,n-2C(i+1)m},
\eenn
and the same bound holds for $ A^t_{m,n,i+1} $. By FKG inequality, we deduce that 
\benn
\phi^{0\backslash 1}_{S_n, p_{sd},q}(A^b_{m,n,i+1} \cap A^t_{m,n,(i+1)} | A^b_{m,n,i} \cap A^t_{m,n,i}) \ge \bigl(\tfrac{1}{2} - p_{m,n-2C(i+1)m} \bigr)^2.
\eenn
Combining it with $\eqref{hugoadd}$, we deduce that 
\benn
\phi^{0\backslash 1}_{S_n, p_{sd},q}(A^b_{m,n,i+1} \cap A^t_{m,n,(i+1)}) \ge \bigl (\tfrac{1}{2} -p_{m,n-2C(i+1)m}\bigr)^{2(i+2)}.
\eenn
Letting $i$ be equal to $k$ gives the result.
\end{proof}

\begin{cor} \label {CorLem3}

Let us fix $C \ge 4$ and $m, n \in \N$ such that
\be \label{relmn}
n \ge 9 C^2 m.
\ee
Let us call $A^b_C$ the event that $[0,m] \times \{0\}$ is connected either to $\{-4Cm-m \} \times [0, 8Cm]$ or to $[-4Cm-m, m] \times \{8Cm\}$ in
 $[-4Cm-m, 4Cm-m] \times [0, 8Cm] \cap L^b_{4C}$.
 
Then,
\be
\phi^{0\backslash 1}_{S_n, p_{sd},q} (A^b_C) \ge (\tfrac{1}{2} - p_{m, C^2m})^{2(4C+1)}.
\ee
\end{cor}

\begin{proof}
The result follows directly from the fact that $p_{m,n}$ is decreasing in the second variable and from the previous lemma applied to $A_{m,n, 4C}$. The built path either ends at $\{-4Cm-m \} \times [0, 8Cm]$ or leaves the box from the top boundary.
\end{proof}

\begin{lemma} \label{contog}
For any choice of $C \ge 4$ and $M > 0$, we have that
\be
\sum _{x \in [-(4C+1)m, -m] \times \{0\}} \phi^{0\backslash 1}_{S_n, p_{sd},q}   (x \arr 0 ) \ge M (\tfrac{1}{2} - p_{m,C^2m})^{2(4C+1)}
\ee
for $m$ large enough and $n\in \N$ such that \eqref{relmn} holds.

\end{lemma}

\begin{proof}

\begin{figure}[b!]
\centering
\includegraphics[width=0.6\textwidth]{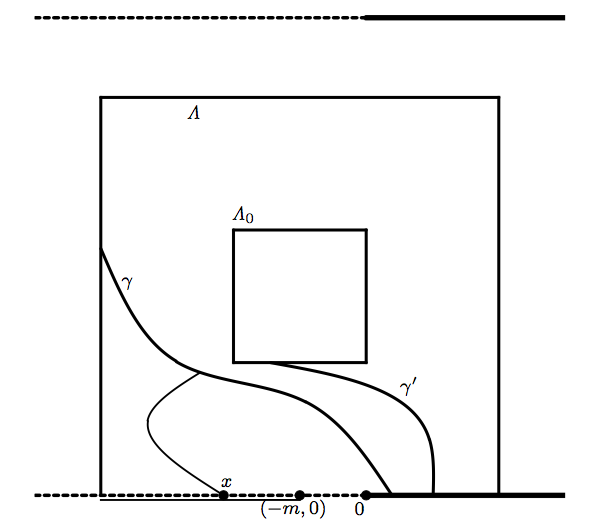}
\caption{Connection of a point $x$ to $\g$ is more probable than its connection to $\g' \cup \mathit{\Lambda}_0$}
\label{figsqinline}
\end{figure}

Let us choose $R$ according to Lemma \ref{inasquare}, $m \ge R$ and $n \ge 9C^2m$ and look at the square box
$\mathit{\Lambda} = \L_{4Cm} + (-m,4Cm)$ with the smaller box $\mathit{\Lambda} ^0 = \L_{R} + (-m,4Cm)$ inside.

Then, conditioned on  $A^b_C$, let $\gamma$ be the uppermost realisation of the path going from $[0,m] \times \{0\}$ to the left or upper parts of $ \d  \mathit{\Lambda} $ in the domain described in Corollary \ref{CorLem3} (see Figure \ref{figsqinline}). Note that $\gamma$ lies in the annulus 
$\mathit{\Lambda} \backslash \mathit{\Lambda}^0$.
Let us choose $\gamma'$ connecting $\d \mathit{\Lambda} $ and $\mathit{\Lambda} ^0$  and going to the right from $\gamma$. Then, by Lemma \ref{inasquare},
\benn
\sum _{x \in  [-(4C+1)m, -m] \times \{0\}} \!\!\!\!\!\!\!\!\!\!\! \phi^{0\backslash 1}_{S_n, p_{sd},q}   (x \arr \gamma ) \ge
\!\!\!\!\!\!\!\!\!\!\! \sum _{x \in  [-(4C+1)m, -m] \times \{0\}} \!\!\!\!\!\!\!\!\!\!\! \phi^{0}_{\mathit{\Lambda}, p_{sd},q}  (x \arr \mathit{\Lambda}^0 \cup \gamma\,|\,
\mathit{\Lambda}^0 \cup \gamma \text{ are open})
\ge M.
\eenn
Increasing the domain to $S_n$ and putting wired boundary conditions on $ \d^+ S_n$ only increases the probabilities in the sum. 
Together with Corollary \ref{CorLem3}, this gives the result.
\end{proof}

Let us now fix any $\e>0$ and define $\rho = \rho(\e)>1$ and $C = C(\rho) >4$ large enough that 
\be \label{Cro1}
(C+1) \le \rho^C
\ee
and
\be \label{Cro2}
C \ge 2\left(\tfrac{\rho}{\rho-1}\right)^2
\ee 
and look at the set
\benn
K_n = \{ k \in \N, k \le \log_\rho  \tfrac{n}{9C^2}, p_{\rho^k (1-\rho),n} \le \tfrac{1}{4}\}.
\eenn
We will prove that we are either in Case 1 or in Case 2 of Lemma \ref{twocases} depending on whether $|K_n|$ is large or not.
\begin{lemma} \label{case1pf}
Fix  $R>0$. For any choice of $\rho >1$ and for any $n$ large enough, if $|K_n| \ge \tfrac{1}{2} \log_\rho \tfrac{n}{9C^2}$, then
\benn
\sum _{x \in [-n,0)} \phi^{0\backslash 1}_{S_n, p_{sd},q}   ((x,0) \arr 0 ) \ge R \log n.
\eenn
\end{lemma}

\begin{proof}
Inequality \eqref{Cro1} together with the condition 
$|K_n| \ge \tfrac{1}{2} \log_\rho  \tfrac{n}{9C^2}$ allows to pick $\tfrac{1}{10C} \log_\rho  \tfrac{n}{9C^2}$ indices $k \in K_n$ such that the intervals
$$I_k = [-(4C+1)\rho^k, -\rho^k]$$
 are disjoint.
Let us choose $M = 20RC 4^{2(4C+1)} $ and apply Lemma \ref{contog} to each $I_k$. Then, we obtain 
\begin{align}
\sum _{x \in [-n,0)} \phi^{0\backslash 1}_{\H, p_{sd},q}   ((x,0) \arr 0 ) &\ge 
\sum_k \sum _{x \in I_k} \phi^{0\backslash 1}_{\H, p_{sd},q}   ((x,0) \arr 0 ) 
\nonumber \\ &\ge \nonumber
\frac{M}{10C 4^{2(4C+1)}} \log_\rho  \tfrac{n}{9C^2} \ge R \log n.
\end{align}
Notice that $R$ is positive and can be picked arbitrary large for $n$ big enough.
\end{proof}

\begin{lemma}
\label{case2lems}
Take $n$ large enough and suppose that $|K_n| \le \tfrac{1}{2} \log_\rho \tfrac{n}{C}$, then  
\benn
\phi^0_{\H,p_{sd}-\e,q} (0 \arr \d \L_n) < n^{-20}.
\eenn
\end{lemma}
\par To prove this theorem we need an additional result.
Let us define the {\it{Hamming distance}}
between two configurations $\omega$ and $\omega'$ by
\benn
H(\omega,\omega') \coloneqq \sum_{e \in E(\Omega)} |\omega(e)-\omega'(e)|,
\eenn
and the Hamming distance between an event $A$ and a configuration $\omega$ ] by
\benn
H_A(\omega) \coloneqq \inf_{\omega' \in A} H(\omega,\omega').
\eenn
Then, we can write an inequality similar to \eqref{difineq} in the terms of the expected Hamming distance \cite{Gr06}:
\be \label{Hamdist}
\frac{d}{dp} \log \phi^\xi_{G,p,q}(A) \ge  \frac{1}{p(1-p)} \phi^\xi_{G,p,q}(H_A).
\ee
We now turn to the proof.

\begin{proof}
Let us study the probability of the following event:
\be
\mathcal{A}_k = 
\left\{[-2\rho^{k+1}, -2 \rho^k]\times \{0\} \arr [2\rho^{k},2\rho^{k+1}]\times \{0\} \text{ in } A_n
\right\}
\ee
where $A_k$ is the half-annulus
$$
A_k = [-2\rho^{k+1},2 \rho^{k+1}]\times [0, C\rho^k(\rho-1)] \backslash [-2\rho^{k}, 2\rho^{k}]\times [0, C\rho^{k-1}(\rho-1)]
$$
(see Figure \ref{fighalfan}).

\begin{figure}[b!]
\centering
\includegraphics[width=0.5\textwidth]{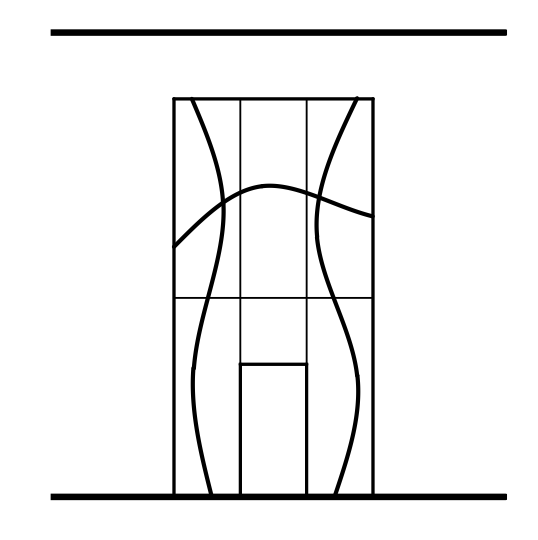}
\caption{Realisation of $\mathcal{A}_k$}
\label{fighalfan}
\end{figure}

\par The rectangles $[-2\rho^{k+1}, - 2\rho^k]\times [0, C\rho^k(\rho-1)]$ and $[2\rho^{k}, 2\rho^{k+1}]\times [0, C\rho^k(\rho-1)]$ are crossed in the vertical direction with probability $p_{\rho^k (1-\rho)}$ each. 
Inequality \eqref{Cro2} implies that $A_k$ contains a square of size $4 \rho^{k+1}$ that is crossed with probability bigger than $\tfrac{1}{2}$. 
Altogether, this gives a bound
$\phi^1_{S_n, p_{sd},q}(\mathcal{A}_k) \ge \tfrac{1}{2} p^2_{\rho^k (1-\rho)}$.
The assumption  on the size of $K_n$ implies that $\phi^1_{S_n, p_{sd},q}(\mathcal{A}_k) \ge \tfrac{1}{32}$ for at least $\tfrac{1}{2C}\log_\rho \tfrac{n}{9C^2}$ values of $k$.
Notice also that for all these indices, $A_k$ lies in $\Lambda_n$ and that the $A_k$ are disjoint.

We study this probability for $p = p_{sd}$ so $\mathcal{A}_k$ implies the existence of dual-open circuit in $\H^*$ disconnecting $0$ from $\d \L_n$. The expected number of such disjoint circuits is bigger than $\tfrac{1}{64}\log_\rho \tfrac{n}{C}$, since the $A_k$ are disjoint. 
This number bounds from below the expectation of the Hamming distance of the event that $0$ is connected to $\d \L_n$ in the dual lattice:
\benn
\phi^{1^*}_{\H, (p_{sd})^*\!\!,q} (H_{0 \overset{*}\arr \d \L_n}) = \phi^0_{\H, p_{sd},q} (H_{0 \arr \d \L_n}) \ge \tfrac{1}{64}\log_\rho \tfrac{n}{C}.
\eenn
Here, we changed to the primal lattice with free boundary conditions using the self-duality of the model for $p_{sd}$. We apply \eqref{Hamdist} to obtain
\benn
\frac{d}{dp} \log \phi^0_{\H, p_{sd},q} [0 \arr \d \L_n] > \frac{1}{64}\log_\rho \frac{n}{C}.
\eenn
By monotonicity, this inequality extends to every $p \in [p_{sd}-\e', p_{sd}]$ for $\e>0$.
Integrating the previous inequality gives:
\benn
\int_{p_{sd}-\e}^{p_{sd} }\frac{d}{dp} \log \phi^0_{\H, p,q} [0 \arr \d \L_n] \ge \frac{\e}{64}\log_\rho\frac{n}{C},
\eenn
or, put differently,
\benn
\phi^0_{\H, p_{sd}-\e,q} [0 \arr \d \L_n] \le \left( \tfrac{C}{n}\right)^{\frac{\e}{64\log \rho}}.
\eenn
The value of $\rho$ was chosen close enough to $1$ that
\benn
\frac{\e}{64\log \rho} >20.
\eenn
This concludes the proof.
\end{proof}

\bibliographystyle{alpha}

\end{document}